\documentclass[11pt]{article}
\usepackage{url,ifthen}
\usepackage{srcltx}
\usepackage{multirow}
\usepackage{boxedminipage}
\usepackage[margin=1.1in]{geometry}
\usepackage{nicefrac}
\usepackage{xspace}
\usepackage{graphicx}
\usepackage{srcltx}
\usepackage[usenames]{color}
\usepackage{fullpage}
\usepackage{xspace}
\definecolor{DarkGreen}{rgb}{0.1,0.5,0.1}
\definecolor{DarkRed}{rgb}{0.5,0.1,0.1}
\definecolor{DarkBlue}{rgb}{0.1,0.1,0.5}
\usepackage[small]{caption}
\usepackage{float}
\usepackage{balance}
\usepackage{amsmath}
\usepackage{amsfonts}
\usepackage{amssymb}
\usepackage{amsthm}
\usepackage{algorithm}
\usepackage{algpseudocode}
\usepackage{algorithmicx}
\usepackage{mathtools}

\newtheorem{theorem}{Theorem}[section]
\newtheorem*{namedtheorem}{\theoremname}
\newcommand{\theoremname}{testing}

\newtheorem{lemma}{Lemma}

\newtheorem{claim}{Claim}

\newtheorem{example}{Example}

\newtheorem{fact}{Fact}

\newtheorem{corollary}{Corollary}

\newtheorem*{question*}{Question}

\theoremstyle{definition}
\newtheorem{definition}{Definition}

\theoremstyle{plain}
\newtheorem{Alg}{Algorithm}

\definecolor{DarkGreen}{rgb}{0.1,0.5,0.1}
\definecolor{DarkRed}{rgb}{0.5,0.1,0.1}
\definecolor{DarkBlue}{rgb}{0.1,0.1,0.5}

\usepackage[pdftex]{hyperref}
\usepackage{cleveref}
\hypersetup{
    unicode=false,          
    pdftoolbar=true,        
    pdfmenubar=true,        
    pdffitwindow=false,      
    pdfnewwindow=true,      
    colorlinks=true,       
    linkcolor=DarkGreen,          
    citecolor=DarkGreen,        
    filecolor=DarkGreen,      
    urlcolor=DarkBlue,          
}

\newcommand{\ignore}[1]{}

\renewcommand{\Pr}{\mathop{\bf Pr\/}}                    

\makeatletter
\floatstyle{ruled}
\newfloat{fragment}{H}{lop}
\floatname{fragment}{Algorithm}
\renewcommand{\floatc@ruled}[2]{\vspace{2pt}{\@fs@cfont \#1.\:} \#2 \par
 \vspace{1pt}}
\makeatother

\author{}

\begin{document}

\title{Efficiently Learning Mixtures of Mallows Models}
\author{Allen Liu \\MIT  \and Ankur Moitra\thanks{This work was supported in part by NSF CAREER Award CCF-1453261, NSF Large CCF-1565235, a David and Lucile Packard Fellowship, and an Alfred P. Sloan Fellowship.}\\ MIT}
\date{}
\maketitle

\begin{abstract}
Mixtures of Mallows models are a popular generative model for ranking data coming from a heterogeneous population. They have a variety of applications including social choice, recommendation systems and natural language processing. Here we give the first polynomial time algorithm for provably learning the parameters of a mixture of Mallows models with any constant number of components. Prior to our work, only the two component case had been settled. Our analysis revolves around a determinantal identity of Zagier \cite{RepTheory} which was proven in the context of mathematical physics, which we use to show polynomial identifiability and ultimately to construct test functions to peel off one component at a time. 

To complement our upper bounds, we show information-theoretic lower bounds on the sample complexity as well as lower bounds against restricted families of algorithms that make only local queries. Together, these results demonstrate various impediments to improving the dependence on the number of components. They also motivate the study of learning mixtures of Mallows models from the perspective of beyond worst-case analysis. In this direction, we show that when the scaling parameters of the Mallows models have separation, there are much faster learning algorithms.  
\end{abstract}

\thispagestyle{empty}
\setcounter{page}{0}

\newpage

\section{Introduction}

\subsection{Background}

User preferences \---- in a wide variety of settings ranging from voting \cite{Voting} to information retrieval \cite{Rank} \---- are often modeled as a distribution on permutations. Here we study the problem of learning a mixture of Mallows models from random samples. First, a Mallows model $M(\phi, \pi^*)$ is described by a center $\pi^*$ and a scaling parameter $\phi$. The probability of generating a permutation $\pi$ is $$\Pr_{M(\phi, \pi^*)}[\pi] = \frac{\phi^{d_{KT}(\pi, \pi^*)}}{Z}$$ where $d_{KT}$ is the Kendall-Tau distance \cite{KT} and $Z$ is a normalizing constant that only depends on $\phi$ and the number of elements being permuted which we denote by $n$. C. L. Mallows introduced this model in $1957$ and gave an inefficient procedure for sampling from them: Rank every pair of elements randomly and independently so that they agree with $\pi^*$ with probability $\frac{1}{1+\phi}$ and output the ranking if it is a total ordering. Doignon et al. \cite{Insertion} discovered a repeated insertion-based model that they proved is equivalent to the Mallows model and more directly lends itself to an efficient sampling procedure. 

The Mallows model is a natural way to represent noisy data when there is one true ranking that correlates well with everyone's own individual ranking. However in many settings (e.g. voting \cite{Gormley}, recommendation systems) the population is heterogenous and composed of two or more subpopulations. In this case, it is more appropriate to model the data as a mixture of simpler models. Along these lines, there has been considerable interest in fitting the parameters of a mixture of Mallows models to ranking data \cite{Pairwise, Mixtures, Dirichlet}. However most of the existing approaches (e.g. Expectation-Maximization \cite{Pairwise}) are heuristic and only recently were the first algorithms with provable guarantees given. For a single Mallows model, Braverman and Mossel \cite{MallowsReconstruction} showed how to learn it by quantifying how close the empirical average of the ordering of elements is to the ordering given by $\pi^*$ as the number of samples increases.

Awasthi et al. \cite{Original} gave the first polynomial time algorithm for learning mixtures of two Mallows models. Their algorithm learns the centers $\pi_1$ and $\pi_2$ exactly and the mixing weights and scaling parameters up to an additive $\theta$ with running time and sample complexity $$\mbox{poly}(n, \frac{1}{w_{min}}, \frac{1}{\phi_1(1-\phi_1)}, \frac{1}{\phi_2(1-\phi_2)}, \frac{1}{\theta} )$$ Here $\phi_1$ and $\phi_2$ are the scaling parameters and $w_{min}$ is the smallest mixing weight. Their algorithm works based on recasting the parameter learning problem in the language of tensor decompositions, similarly to other algorithms for learning latent variable models \cite{Tensors}. However there is a serious complication in that most of the entries in the tensor are exponentially small. So even though we can compute unbiased estimates of the entries of a tensor whose low rank decomposition would reveal the parameters of the Mallows model, most of the entries cannot be meaningfully estimated from few samples. Instead, Awasthi et al. \cite{Original} show how the entries that can be accurately estimated can be used to learn the prefixes of the permutations, which can be bootstrapped to learn the rest of the parameters. In fact before their work, it was not even known whether a mixture of two Mallows models was {\em identifiable} \---- i.e. whether its parameters can be uniquely determined from an infinite number of samples. 

The natural open question was to give provable algorithms for learning mixtures of any constant number of Mallows models. For other learning problems like mixtures of product distributions \cite{FM, FOS} and mixtures of Gaussians \cite{KMV, MV, BS}, algorithms for learning mixtures of two components under minimal conditions were eventually extended to any constant number of components. Chierichetti et al. \cite{Chier} showed that when the number of components is exponential in $n$, identifiability fails. On the other hand, when all the scaling parameters are the same and known, Chierichetti et al. \cite{Chier} showed that it is possible to learn the parameters when given an arbitrary (and at least exponential in $n \choose 2$) number of samples. Their approach was based on the Hadamard matrix exponential. They also gave a clustering-based algorithm that runs in polynomial time and works whenever the centers are well-separated according to the Kendall-Tau distance, by utilizing recent concentration bounds for Mallows models that quantify how close a sample $\pi$ is likely to be to the center $\pi^*$ \cite{Lengths}. 

\subsection{Our Results and Techniques}

Our main result is a polynomial time algorithm for learning mixtures of Mallows models for any constant number of components. Let $d_{TV}$ be the total variation distance, let $w_{min}$ be the smallest mixing weight, and let $U$ denote the uniform distribution over the $n!$ possible permutations. We prove:

\begin{theorem}\label{generalalgorithm}
For any constant $k$, given samples from a mixture of $k$ Mallows models $$M = w_1 M(\phi_1, \pi_1) + \cdots + w_k M(\phi_k, \pi_k)$$ where $d_{TV}(M(\phi_i, \pi_i), M(\phi_j, \pi_j)) \geq \mu$ for all $i \neq j$, $d_{TV}(M(\phi_i, \pi_i),U) \geq \mu$ for all $i$ and $n \geq 10 k^2$, there is an algorithm whose running time and sample complexity are $$\mbox{poly}(n, \frac{1}{w_{min}}, \frac{1}{\mu}, \frac{1}{\theta}, \log \frac{1}{\delta} )$$ for learning each center $\pi_i$ exactly and the mixing weights and scaling parameters to within an additive $\theta$. Moreover the algorithm succeeds with probability at least $1-\delta$.
\end{theorem}

A main challenge in learning mixtures of Mallows models is in establishing {\em polynomial identifiability} \---- i.e. that the parameters of the model can be approximately determined from a polynomial number of samples. When addressing this question, there is a natural special matrix $A$ to consider: Let $A$ be an $n! \times n!$ matrix whose rows and columns are indexed by permutations with $$A_{\pi, \sigma} = \phi^{d_{KT}(\pi, \sigma)}$$ Zagier \cite{RepTheory} used tools from representation theory to find a simple expression for the determinant of $A$. Interestingly his motivation for studying this problem came from interpolating between Bose and Fermi statistics in mathematical physics. We can translate his result into our context by observing that the columns of $A$, after normalizing so that they sum to one, correspond to Mallows models with the same fixed scaling parameter $\phi$. Thus Zagier's result implies that any two distinct mixtures $M$ and $M'$ of Mallows models, where all the components have the same scaling parameter, produce different distributions\footnote{This result was rediscovered by Chierichetti et al. \cite{Chier} using different tools, but without a quantitative lower bound on the smallest singular value.}. 

However the quantitative lower bounds that follow from Zagier's expression for the determinant are too weak for our purposes, and are not adapted to the number of components in the mixture. We exploit symmetry properties to show lower bounds on the length of any column of $A$ projected onto the orthogonal complement of any other $k-1$ columns, which allows us to show that not only does $A$ have full rank, any small number of its columns are robustly linearly independent \cite{Allman}. More precisely we prove:

\begin{theorem}
Let $\phi < 1 - \epsilon$. Let $c_1, c_2, \cdots, c_k$ be any $k$ distinct columns of $A$, normalized so that they each sum to one. Then $$\|z_1 c_1 + \cdots + z_k c_k\|_{1} \geq \frac{\max_i(|z_i|) \epsilon^{2k^2}}{2n^{4k} (k+1)^{2k^2 + 4k}} $$
where $z_i$ are arbitrary real coefficients. 
\end{theorem}

Even though this result nominally applies to mixtures of Mallows models where all the scaling parameters are the same, we are able to use it as a black box to solve the more general learning problem. We reformulate our lower bound on how close any $k$ columns can be to being linearly dependent in the langugage of test functions, which we use to show that when the scaling parameters are different, we can isolate one component at a time and subtract it off from the rest of the mixture. Combining these tools, we obtain our main algorithm.  We note that the separation conditions we impose between pairs of components are information-theoretically necessary for our learning task.  

It is natural to ask whether the dependence on $k$ can be improved. First, we show lower bounds on the sample complexity. We construct two mixtures $M$ and $M'$ whose components are far apart \---- every pair of components has total variation distance at least $\mu$ \---- but $M$ and $M'$ have total variation distance about $\mu^{2k-1}$. As a corollary we have:

\begin{corollary}
Any algorithm for learning the components of a mixture of $k$ Mallows models within $\mu$ in total variation distance must take at least $(1/\mu)^{2k-1}$ samples. 
\end{corollary}

Second, we consider a restricted model where the learner can only make local queries of the form: Given elements $x_1, \cdots, x_c$ and locations $i_1, \cdots, i_c$ and a tolerance $\tau$, what is the probability that the mixture assigns $x_j$ to location $i_j$ for all $j$ from $1$ to $c$, up to an additive $\tau$? We show that our algorithms can be implemented in the local model. Moreover, in this model, we can prove lower bounds on the dependence on $n$ and $k$. We show:

\begin{theorem}[Informal]
Any algorithm for learning a mixture of $k$ Mallows models through local queries must make at least $n^{\log k}$ queries or make a query with $\tau \leq n^{-\frac{1}{2} \log k}$. 
\end{theorem}

This is reminiscent of statistical query lower bounds for other unsupervised learning problems, most notably learning mixtures of Gaussians \cite{DKK}. However it is not clear how to prove lower bounds on the statistical query dimension \cite{SQL}, because of the complicated ways that the locations that each element is mapped to affect one another in a Mallows model, which makes it challenging to embed small hard instances into larger ones. 

Finally we turn to beyond-worst case analysis and ask whether there are natural conditions on the mixture that allow us to get algorithms whose dependence on $n$ is a fixed polynomial rather than one whose degree depends on $k$. Rather than requiring the centers to be far apart, we merely require their scaling parameters to be separated from one-another. We show:

\begin{theorem}\label{mainsmooth}
Given samples from a mixture of $k$ Mallows models $$M = w_1 M(\phi_1, \pi_1) + \cdots + w_k M(\phi_k, \pi_k)$$ where $|\phi_i - \phi_j| \geq \gamma$ for all $i \neq j$, $\phi_i \leq 1 - \gamma$ for all $i$ and $n \geq 10 k$, there is an algorithm whose running time and sample complexity are $$f(\gamma, \theta, w_{\min}, k) \mbox{poly}(n, \log \frac{1}{\delta})$$ for learning each center $\pi_i$ exactly and the mixing weights and scaling parameters to within an additive $\theta$, where $f(\gamma, \theta, w_{\min}, k) = \mbox{poly}(1/\gamma^{k^2}, 1/\theta^{k^2}, 1/w_{min}^{k^2})$. Moreover the algorithm succeeds with probability at least $1-\delta$.
\end{theorem}

Our algorithm leverages many of the lower bounds on the total variation distance between mixtures of Mallows models and test functions for separating one component from the others that we have established along the way. 

\subsection*{Further Related Work}

There are other natural models for distributions on permutations such as the Bradley-Terry model \cite{Bradley} and the Plackett-Luce model \cite{PlackettO, Luce}. Zhao et al. \cite{Plackett} showed that a mixture of $k$ Plackett-Luce models is generically identifiable provided that $k \leq \lfloor \frac{n-2}{2} \rfloor!$ and gave a generalized method of moments algorithm that they proved is {\em consistent} \---- meaning that as the number of samples goes to infinity, the algorithm recovers the true parameters. More generally, Teicher \cite{Tei61, Tei63} obtained sufficient conditions for the identifiability of finite mixtures but these conditions do not apply in our setting. 

\section{Preliminaries}
\subsection{Basic Notation}\label{sec:notation}
Let $[n] = \{1,2, \cdots, n \}$.  Given two permutations $\pi$  and $\pi'$ on $[n]$, let $d_{KT}(\pi, \pi')$ denote the Kendall-Tau distance, which counts the number of pairs $(i,j)$ for which the two rankings disagree. 

\begin{definition}
A Mallows model $M(\phi , \pi^*, n)$ defines a distribution on permutations of the set $[n]$ where the probability of generating a permutation $\pi$  is equal to $$\frac{\phi^{d_{KT}(\pi^*, \pi)}}{Z_{n}(\phi)}$$
and $Z_n(\phi) = \sum_{\pi} \phi^{d_{KT}(\pi^*,\pi)}$ be the normalizing constant, which is easy to see is independent of $\pi^*$. When the number of elements is clear from context, we will omit $n$ and write $M(\phi, \pi^*)$. 
\end{definition}

The following is a well-known (see e.g. \cite{Insertion}) iterative process for generating a ranking from $M(\phi, \pi^*)$: Consider the elements in rank decreasing order, according to $\pi$.  When we reach the $(i+1)^{st}$ ranked element, it is inserted into each of the $i+1$ possible positions with probabilities $$\frac{\phi^i}{1 + \phi + \cdots + \phi^i}, \cdots ,\frac{\phi}{1 + \phi + \cdots + \phi^i}, \frac{1}{1 + \phi + \cdots + \phi^i}$$ respectively, where the order of the probabilities go from the highest rank position it could be inserted to the lowest. When the last element is inserted, the result is a random permutation drawn from $M(\phi, \pi^*)$.

A mixture of $k$ Mallows models is defined in the usual way: We write $$M = w_1M(\phi_1, \pi_1^*) + \cdots + w_kM(\phi_k, \pi_k^*)$$ where the mixing weights $w_1, w_2, \cdots, w_k$ are nonnegative and sum to one. A permutation is generated by first choosing an index (each $i$ is chosen with probability $w_i$) and then drawing a sample from the corresponding Mallows model $M(\phi_i, \pi_i^*)$.

We will often work with the natural vectorizations of probability distributions: 

\begin{definition}
If $P$ is a distribution over permutations on $[n]$ we let $v(P)$ denote the length $n!$ vector whose entries are the probabilities of generating each possible permutation. We will abuse notation and write $v(M)$ for the vectorization of a Mallows model $M$. 
\end{definition}

Our algorithms and their analyses will frequently make use of the notion of restricting a permutation to a set of elements:
\begin{definition}
Given a permutation $\pi$ on $[n]$ and a subset $S \subseteq [n]$, let $\pi\vert_S$ be the permutation on the elements of $S$ induced by $\pi$.
\end{definition}

\subsection{Block and Orders}
Our algorithms will be built on various structures we impose on permutations. The way to think about these structures is that each one gives us a statistic that we can measure: What is the probability that a permutation sampled from an unknown Mallows model has the desired structure? These act like natural moments of the distribution, that we will manipulate and use in conjunction with tensor methods to design our algorithms.

\begin{definition}
A block structure $\mathcal{B} = S_1, S_2, \cdots , S_j$ is an ordered collection of disjoint subsets of $[n]$. We say that a permutation $\pi$ satisfies $\mathcal{B}$ as a block structure if for each $i$, the elements of $S_i$ occur consecutively (i.e. in positions $a_i,a_i+1, \dots, a_i + |S_i|-1$ for some $a_i$) in $\pi$ and moreover the blocks occur in the order $S_1,S_2, \cdots, S_j$. Finally we let $\mathcal{S}_\mathcal{B}$ denote the set of permutations satisfying $\mathcal{B}$ as a block structure. 
\end{definition}

\begin{definition}
An order structure $\mathcal{O} = S_1, S_2, \cdots, S_j$ is a collection of ordered subsets of $[n]$.  We say a permutation $\pi$ satisfies $\mathcal{O}$ as an order structure if for each $i$, the elements of $S_i$ occur in $\pi$ in the same relative order as they do in $S_i$. 
\end{definition}

\begin{definition}
An ordered block structure $\mathcal{A} = S_1, S_2, \cdots, S_j$ is an ordered collection of ordered disjoint subsets of $[n]$.  We say a permutation $\pi$ satisfies $\mathcal{A}$ as an ordered block structure if it satisfies $ S_1, S_2, \cdots, S_j $ both as a block structure and as an order structure \---- i.e. we forget the order within each $S_i$ when we treat it as a block structure and we forget the order among the $S_i$'s when we treat it as an order structure. 
\end{definition}

To help parse these definitions, we include the following example:

\begin{example}
Let $n=7$ and consider $\mathcal{A} = (1,2), (4,5,6) $.  The permutation $(1,2,3,7,6,5,4)$ satisfies $\mathcal{A}$ as a block structure.  The permutation $(1,3,4,2,5,6,7)$ satisfies $\mathcal{A}$ as an order structure and the permutation $(1,2,3,4,5,6,7)$ satisfies $\mathcal{A}$ as an ordered block structure.  
\end{example}

\section{Basic Facts}

Here we collect some basic facts about Mallows models, in particular a lower bound on the probability that they satisfy a given block structure if their base permutation does, relationships between the total variation distance and parameter distance, and determinantal identities for special matrices. 

\subsection{What Block Structures are Likely to be Satisfied?}

In this subsection, our main result is a lower bound on the probability that a permutation drawn from a Mallows model satisfies a block structure that the underlying base permutation does. Along the way, we will also establish some useful ways to think about conditioning and projecting Mallows models in terms of tensors. 

\begin{fact}\label{block}
The conditional distribution of a Mallows model $M(\phi,\pi^*)$ when restricted to rankings where the elements in the set $S$ (of size $j$) are ranked in positions $a,a+1, \cdots, a+j-1$ and the ranking of elements in $[n] \setminus S$ is fixed is precisely $M(\phi, \pi\vert_S^*)$.
\end{fact}
\begin{proof}
It is easy to see that for any two permutations $\tau$ and $\tau'$ on $[n]$ where the elements of $S$ are ranked in positions $a,a+1, \cdots, a+j-1$ and agree on the rankings of elements in  $[n] \setminus S$ satisfy
$$d_{KT}(\pi^*,\tau) - d_{KT}(\pi^*,\tau') = d_{KT}(\pi\vert_S^*,\tau\vert_S) - d_{KT}(\pi\vert_S^*,\tau'\vert_S)$$
Thus the ratio of their probabilities is the same as the ratio of probabilities of $\tau\vert_S$ and $\tau'\vert_S$ in $M(\phi, \pi\vert_S^*)$, which completes the proof.
\end{proof}

Next we will describe a natural way to think about the conditional distribution on permutations that satisfy a given block structure as a tensor. Recall that the subsets of $[n]$ in a block structure are required to be disjoint. 

\begin{definition}\label{def:tensor}
Given a Mallows model $M(\phi, \pi^*)$ and a block structure $\mathcal{B} = S_1, S_2, \cdots , S_j$, we define a 
$ |S_1|! \times |S_2|! \times \cdots \times |S_j|!$
dimensional tensor $T_{M, \mathcal{B}}$ as follows: Each entry corresponds to orderings $\pi_1, \pi_2, \cdots, \pi_j$ of $S_1, S_2, \cdots, S_j$ respectively and in it, we put the probability that a ranking drawn from $M$ satisfies $\mathcal{B}$ and for each $i$, the elements in $S_i$ occur in the order specified by $\pi_i$. 
\end{definition}

It is easy to see that $T_{M, \mathcal{B}}$ has rank one. Technically this requires the obvious generalization of Fact~\ref{block} where we condition on the elements in each $S_i$ occurring in specified consecutive locations, and then note that these events are all disjoint. 

\begin{corollary}\label{blockStructure}
$T_{M, \mathcal{B}} = \Pr_M[\pi \in \mathcal{S}_\mathcal{B}] \cdot v(M(\phi,\pi\vert_{S_1})) \otimes \dots \otimes v(M(\phi,\pi\vert_{S_j})) $
\end{corollary}

Our next result gives a convenient lower bound for the probability that a sample satisfies a given block structure $B$ provided that the base permutation satisfies $B$:  

\begin{lemma}\label{fix}
For any Mallows model $M(\phi, \pi^*)$ and block structure $\mathcal{B} = S_1, S_2, \cdots , S_j$ where $\pi^*$ satisfies $\mathcal{B}$ and $\ell = |S_1|+ \cdots + |S_j|$, we have $$\Pr_M[\pi \in \mathcal{S}_\mathcal{B}] \geq \frac{1}{n^{2\ell}}$$
\end{lemma}
\begin{proof}
Without loss of generality let $\pi = (1,2, \dots, n)$ and $S_i = [a_i, a_i+b_i-1]$.  We will first consider the set of permutations $\mathcal{T}_{\mathcal{B}}$ with the following property: For each $1\leq i \leq j$, the elements $a_i,a_i+1, \cdots a_i+b_i-1$ occur in their natural order and each of them occurs after all of the elements $1,2, \cdots a_i-1$.  Using the iterative procedure for sampling from a Mallows model defined in Section~\ref{sec:notation} and the fact that $\frac{1}{1 + \phi + \cdots + \phi^i} \geq \frac{1}{n}$, we have that $$\Pr_M[\pi \in \mathcal{T}_\mathcal{B}] \geq \frac{1}{n^{\ell}}$$
since we just need that when each element in each $S_i$ is inserted, it is inserted in the lowest ranked position available. 

Next we build a correspondence between permutations in $\mathcal{T}_\mathcal{B}$ and permutations in $\mathcal{S}_\mathcal{B}$.  Consider $\tau \in \mathcal{T}_\mathcal{B}$. For each block $a_i,a_i+1, \cdots, a_i+b_i -1$, say that in $\tau$, they are placed in positions $$x_{i1} < x_{i2} < \cdots < x_{ib_i}$$ Note that these are not necessarily consecutive. We will map $\tau$ to a permutation in $S_B$ by making them consecutive while preserving their order by doing the following for each block: Place $a_i,a_i+1, \cdots, a_i+b_i-1$ in positions $x_{i1}, \cdots, x_{i1}+b_i-1$ and all other elements displaced from their position are placed afterwards while preserving their ordering. Crucially this process only reduces the number of inversions in $\tau$ due to the way that $\mathcal{T}_\mathcal{B}$ was defined.  In particular, the permutation $\tau'$ we get from this process has the property that it is at least as likely as $\tau$ to be generated from $M$. Also note the intervals $[x_{11}, x_{1b_1}], \cdots, [x_{j1},x_{jb_j}]$ must be disjoint and occur in that order. Now in this correspondence, the order of all elements outside $S_1 \cup \cdots \cup S_j$ is preserved.  Thus, there are at most $n^\ell$ different permutations in $\mathcal{T}_\mathcal{B}$ that can be mapped to the same element in $\mathcal{S}_\mathcal{B}$.  Putting this all together implies that 
$$\Pr_M[\pi \in \mathcal{S}_\mathcal{B}] \geq \frac{\Pr_M[\pi \in \mathcal{T}_\mathcal{B}]}{n^{\ell}} \geq \frac{1}{n^{2\ell}}$$
which completes the proof.
\end{proof}

\subsection{Total Variation Distance Bounds} \label{TVdistanceBounds}
In this subsection, we give some useful relationships between the total variation distance and the parameter distance between two Mallows models (in special cases) in terms of the distance between their base permutations and scaling parameters. We will defer the proofs to Appendix \ref{appendix}.  First we prove that if two Mallows models have different base permutations and their scaling parameters are bounded away from one, then the distributions that they generate cannot be too close.

\begin{claim} \label{simple}
Consider two Mallows models $M_1 = M(\phi_1, \pi_1)$ and $M_2 = M(\phi_2, \pi_2)$ where $\pi_1 \neq \pi_2$ and $\phi_1, \phi_2 \leq 1-\epsilon$. Then $d_{TV}(M_1, M_2) \geq \frac{\epsilon}{2}$.
\end{claim}

Second, we give a condition under which we can conclude that two Mallows models are close in total variation distance. An analogous result is proved in \cite{Original} (see Lemma 2.6) except that here we remove the dependence on $\phi_{min}$.

\begin{lemma}\label{TVbound}
Consider two Mallows models $M_1 = M(\phi_1, \pi)$ and $M_2 = M(\phi_2, \pi)$ with the same base permutation on $n \geq 2$ elements.  If $|\phi_1-\phi_2| \leq \frac{\mu^2}{10n^3}$ then $d_{TV}(M_1, M_2) \leq \mu$.
\end{lemma}

\subsection{Special Matrix Results} \label{matrix}
Here we present a determinantal identity from mathematical physics that will play a central role in our learning algorithms. 
Note $d_{KT}(\pi,\sigma) = I(\pi\sigma^{-1})$ where $I$ counts the number of inversions in a permutation. We make the following definition.
\begin{definition}
Let $A_n(\phi)$ be the $n! \times n!$ matrix whose rows and columns are indexed by permutations $\pi, \sigma$ on $[n]$ and whose entries $A_{\pi\sigma}$ are $\phi^{I(\pi\sigma^{-1})}$.
\end{definition}
 
\noindent Zagier \cite{RepTheory} gives us an explicit form for the determinant of this matrix, which we quote here:
\begin{theorem}\label{thm:Zagier}\cite{RepTheory}
$\det(A_n(\phi)) = \prod_{i=1}^{n-1}(1-\phi^{i^2+i})^{\frac{n!(n-i)}{i^2+i}}$
\end{theorem}

This expression for the determinant gives us weak lower bounds on the total variation distance between mixtures of Mallows models where all the scaling parameters are the same. We will bootstrap this identity to prove a stronger result about how far a column of $A$ is from the span of any set of $k-1$ other columns. 


\section{Identifiability}


In this section we show that any two mixtures of $k$ Mallows models whose components are far from each other (and the uniform distribution) in total variation distance are far from each other as mixtures too, provided that $n > 10k^2$.

\subsection{Robust Kruskal Rank}

Our first step is to show that any $k$ columns of $A_n(\phi)$ are not too close to being linearly dependent \---- i.e. the projection of any column onto the orthogonal complement of the span of any $k-1$ other columns cannot be too small. The {\em Kruskal rank} of a collection of vectors is the largest $\ell$ so that every $\ell$ vectors are linearly independent. The property we establish here is sometimes called a {\em robust Kruskal rank} \cite{Allman}.

\begin{lemma}\label{preK}
Suppose $\phi< 1-\epsilon$ and consider $k$ columns of $A_n(\phi)$.  The projection of one column onto the orthogonal complement of the other $k-1$ has euclidean length at least $(\frac{\epsilon^{n}}{\sqrt{n!}})^k$.
\end{lemma}
\begin{proof}
Assume for the sake of contradiction that there is a set of $k$ columns that violates the statement of the lemma. In particular suppose that the projection of $c_k$ onto the orthogonal complement of $c_1, c_2, \cdots, c_{k-1}$ has euclidean length $N$ for some $N < (\frac{\epsilon^{n}}{\sqrt{n!}})^k$. We will use this assumption to prove an upper bound on the determinant of $A_n(\phi)$ that violates Theorem~\ref{thm:Zagier}. Our approach is to find an ordering of the columns so that as we scan through, at least once every $k$ columns the euclidean length of its projection onto the orthogonal complement of the columns seen so far is at most $N$. Then using the naive upper bound of $\sqrt{n!}$ on the euclidean length of any column of $A_n(\phi)$ we have
$$ \det(A_n(\phi))  \leq N^{\frac{n!}{k}} (\sqrt{n!})^{n!}$$
However from Theorem~\ref{thm:Zagier} we have
$$\det(A_n(\phi)) = \prod_{i=1}^{n-1}(1-\phi^{i^2+i})^{\frac{n!(n-i)}{i^2+i}} \geq (1-\phi)^{n!n(\sum_{i=1}^{n-1}\frac{1}{i(i+1)})} \geq (1-\phi)^{n! n} \geq \epsilon^{n! n}$$
which yields the desired contradiction. 

Now we complete the argument by constructing the desired ordering of the columns as follows: We start with $c_1, c_2, \cdots, c_{k}$. And then we choose any column $c$ not yet selected. Let $\pi$ be the permutation that maps $c_k$ to $c$. Now $\pi$ maps $c_1, c_2, \cdots, c_{k-1}$ to $k-1$ columns, and suppose $j$ of them have not been selected yet. Call these $c'_1, c'_2, \cdots, c'_j$. We continue the ordering of the columns by appending $c'_1, c'_2, \cdots, c'_j, c$. It is easy to see that the euclidean length of the projection of $c$ onto the orthogonal complement of the columns seen so far is also at most $N$, which now finishes the proof. 
\end{proof}

The above lemma is not directly useful for two reasons: First, the lower bound is exponentially small in $n$. Second, it is tantamount to a lower bound on the $\ell_2$-norm of any sparse linear combination of the columns of $A_n(\phi)$. What we really want in the context of identifiability is a lower bound on the $\ell_1$-norm (of a matrix whose columns represent the components). 

\begin{definition}
Let $B_n(\phi)$ be obtained from $A_n(\phi)$ by normalizing its columns to sum to one. 
\end{definition}

\begin{lemma}\label{Kruskal}
Suppose $\phi< 1-\epsilon$ and consider any $k$ columns $c_1,c_2, \dots, c_k$ of $B_n(\phi)$.  Then
$$\|z_1c_1 + \cdots + z_kc_k\|_1 \geq \frac{1}{n^{4k}}\frac{\epsilon^{2k^2}}{(k+1)^{k^2+2k}}$$
provided that $\max(|z_1|,|z_2|, \dots, |z_k|) \geq 1$. 
\end{lemma}
\begin{proof} Let  $\pi_1, \pi_2, \cdots,  \pi_k$ be the permutations corresponding to the columns $c_1,c_2, \dots, c_k$. Also without loss of generality suppose $z_1 \geq 1$ and that $\pi_1 = (1,2, \cdots, n)$. First we build a block structure that $\pi_1$ satisfies but no other $\pi_i$ does: For each $i$, pick two consecutive elements in $\pi_1$, say $x_i$ and $x_i+1$ that are inverted in $\pi_i$. Such a pair exists because $\pi_1 \neq \pi_i$. Now we can take the union of these pairs over all $i$ to form a block structure $ \mathcal{B} = \{ S_1,S_2, \cdots, S_j \}$ so that, for all $i$, $x_i$ and $x_i+1$ are in the same block and $\pi_1$ satisfies $\mathcal{B}$. Note that $j$ can be less than $k$, if for example two of the pairs contain the same element. In any case, we have $|S_1|+ \cdots + |S_j| \leq 2k$. 

Now for each $i$, set $M_i =M(\phi, \pi_i)$ and   $T_i = T_{M_i, \mathcal{B}}$. From Corollary~\ref{blockStructure} we have that
$$T_i = \Pr_{M_i}[\pi \in \mathcal{S}_\mathcal{B}] \cdot v(M(\phi,\pi_i\vert_{S_1})) \otimes \dots \otimes v(M(\phi,\pi_i\vert_{S_j})) $$
Next we show  that we can find unit vectors $v_1, \cdots, v_j$ so that 
\begin{enumerate}

\item[(1)] $\langle v_b, v(M(\phi, \pi_i \vert_{S_b})) \rangle = 0$ whenever $\pi_i \vert_{S_b} \neq \pi_1 \vert_{S_b} $ and

\item[(2)] $\langle v_b, v(M(\phi, \pi_1 \vert_{S_b})) \rangle \geq \frac{1}{|S_b|!}(\frac{\epsilon^{|S_b|}}{\sqrt{|S_b|!}})^k$ for all $b$
\end{enumerate}
This fact essentially follows from Lemma~\ref{preK}. First observe that the $v(M(\phi, \pi_i \vert_{S_b}))$ and the column of  $A_{|S_b|}(\phi)$ corresponding to $\pi_i \vert_{S_b}$ differ only by a normalization, since the former sums to one. Now for each $b$ we can take $v_b$ to be the unit vector in the direction of the projection of $v(M(\phi, \pi_1 \vert_{S_b}))$ onto the orthogonal complement of all the $v(M(\phi, \pi_i \vert_{S_b}))$'s for $i \neq 1$. Note that the additional $\frac{1}{|S_b|!}$ factor arises because Lemma~\ref{preK} deals with $A_{|S_b|}(\phi)$ and to normalize any column of it we need to divide by at most $|S_b|!$. 

With this construction, we have that $\langle v_1 \otimes v_2 \otimes \cdots \otimes v_j, T_i \rangle = 0$ for all $i \neq 1$ since each $\pi_i$ differs from $\pi_1$ when restricted to at least one of the blocks $S_1,S_2, \cdots, S_j$.  Moreover using property $(2)$ above and Lemma~\ref{fix} we have
$$\langle v_1 \otimes v_2 \otimes \cdots \otimes v_j, T_1 \rangle \geq \frac{1}{n^{4k}}\prod_b \frac{1}{|S_b|!}(\frac{\epsilon^{|S_b|}}{\sqrt{|S_b|!}})^k \geq \frac{1}{n^{4k}}\frac{1}{(2k)!}(\frac{\epsilon^{2k}}{\sqrt{(2k)!}})^k \geq \frac{1}{n^{4k}}\frac{\epsilon^{2k^2}}{(k+1)^{k^2+2k}}$$
where the last inequality follows from the bound $(2k)! \leq (k+1)^{2k}$. 
Finally note that $$\|z_1c_1 + \cdots + z_kc_k\|_1 \geq \sum_i \langle v_1 \otimes v_2 \otimes \cdots \otimes v_j, T_i \rangle$$
since the entries of $v_1 \otimes v_2 \otimes \cdots \otimes v_j$ are at most one in absolute value, and each $T_i$ can be formed from $c_i$ by zeroing out entries (corresponding to permutations that do not satisfy $\mathcal{S}_\mathcal{B}$) and summing subsets of the remaining ones together (that correspond to permutations with the same ordering for each $S_i$). 
\end{proof}

The above lemma readily implies that any two mixtures of $k$ Mallows models whose components all have the same scaling parameter and whose mixing weights are different are far from each other in total variation distance. In the sequel we will be interested in proving identifiability even when the scaling parameters are allowed to be different. As a step towards that goal, first we give a simple extension that allows the scaling parameters to be slightly different:

\begin{lemma}\label{RobustKruskal}
Consider any $k$ distinct permutations $\pi_1, \pi_2, \cdots, \pi_k$ and scaling parameters $\phi_1, \phi_2, \cdots, \phi_k$. Let $c_i = v(M(\phi_i, \pi_i))$ and suppose that for each $i$, $\phi_i \leq 1-\epsilon$ and for each $i \neq j$, 
$$|\phi_i -\phi_j| \leq \frac{1}{160n^{8k+3}}\frac{\epsilon^{4k^2}}{(k+1)^{2k^2+4k + 2}}$$
Then for any coefficients $z_i$ with  $\max(|z_1|,|z_2|, \dots, |z_k|) \geq 1$ we have $$\|z_1c_1 + \cdots + z_kc_k\|_1 \geq \frac{1}{2n^{4k}}\frac{\epsilon^{2k^2}}{(k+1)^{k^2+2k}}$$ 
\end{lemma}
\begin{proof}
We can assume $\max(|z_1|,|z_2|, \dots, |z_k|) = 1$ since otherwise we can just scale the $z_i$.  Set $\phi = \phi_1$, and for each $i$ set $c'_i = v(M(\phi, \pi_i))$. Using Lemma~\ref{TVbound} we have
$$\|c_i - c'_i\|_1 = 2 d_{TV}(M(\phi_i, \pi_i), M(\phi, \pi_i)) \leq  \frac{1}{2n^{4k }}\frac{\epsilon^{2k^2}}{(k+1)^{k^2+2k + 1}}$$
And finally invoking Lemma~\ref{Kruskal} we have that
$$\|z_1c_1 + \cdots + z_kc_k\|_1 \geq \|z_1c'_1 + \cdots + z_kc'_k\|_1 - \sum_i \|c_i - c'_i\|_1 \geq \frac{1}{2n^{4k}}\frac{\epsilon^{2k^2}}{(k+1)^{k^2+2k}}$$
which completes the proof. 
\end{proof}

\subsection{Polynomial Identifiability}\label{ident}

Now we are ready to prove our main identifiability result. First we define a notion of non-degeneracy, which is information-theoretically necessary when our goal is to identify all the components in the mixture.

\begin{definition}
We say a mixture of Mallows models
$M = w_1 M(\phi_1,\pi_1) + \cdots + w_k M(\phi_k,\pi_k)$
is $\mu$-non degenerate if the total variation distance between any pair of components is at least $\mu$ and the total variation distance between any component and the uniform distribution over all permutations is also at least $\mu$.
Furthermore we say that the mixture is $(\mu, \alpha)$-non degenerate if in addition each mixing weight is at least $\alpha$. 
\end{definition}

We will not need the following definition until later (when we state the guarantees of various intermediary algorithms), but let us also define a natural notion for two mixtures to be component-wise close:

\begin{definition}
We say that two mixtures of Mallows models
$M = w_1 M(\phi_1,\pi_1) + \cdots + w_k M(\phi_k,\pi_k)$ and $M' = w'_1 M(\phi'_1,\pi'_1) + \cdots + w'_k M(\phi'_k,\pi'_k)$ with the same number of components are component-wise $\theta$-close if there is a relabelling of components in one of the mixtures after which $|w_i - w'_i|, |\phi_i - \phi'_i|$ and $d_{TV}(M(\phi_i,\pi_i), M(\phi'_i,\pi'_i)) \leq \theta$ for all $i$. 
\end{definition}

Most of these conditions are standard, because in order to be able to identify $M$ from a polynomial number of samples we need to get at least one sample from each component and at least one sample from the difference between any two components. In our context, we additionally require no component to be too close to the uniform distribution because the distribution $M(1, \pi)$ is the same regardless of the choice of $\pi$. 

Below, we state and prove our main lemma in this section. The technical argument is quite involved, but many of the antecedents (in particular finding block structures that capture the disagreements between permutations, representing the distribution on a subset of permutations as a tensor and constructing test functions as the tensor product of simple vectors) were used already in the proof of Lemma~\ref{Kruskal}.

\begin{lemma}\label{Identifiability}
Consider any $k$ (not necessarily distinct) permutations $\pi_1, \pi_2, \cdots, \pi_k$ and scaling parameters $\phi_1, \phi_2, \cdots, \phi_k$. Set $M_i = M(\phi_i, \pi_i)$ and suppose that the collection of Mallows models is $\mu$-non degenerate. Then for any coefficients $z_i$ with  $\max(|z_1|,|z_2|, \dots, |z_k|) \geq 1$ we have $$\|z_1v(M_1) + \cdots + z_k v(M_k)\|_1 \geq \Big(\frac{\mu^2}{10n^4 k}\Big)^{20k^3}$$ 
\end{lemma}

Again, it suffices to consider $\max(|z_1|,|z_2|, \dots, |z_k|) = 1$.  Set $\epsilon = \frac{\mu^2}{10n^3}$.  We will break up the proof of this lemma into two cases. First consider the case where $\phi_1 \geq \frac{\epsilon}{2nk}$ and $|z_1|> \frac{1}{k}$. We will use the following intricate construction: Consider the set of $\phi_i$ such that $$|\phi_i - \phi_1| > \frac{1}{160n^{8k+5}}\frac{\epsilon^{4k^2}}{(k+1)^{2k^2+4k}}$$ and suppose without loss of generality these are $\phi_{j+1}, \cdots, \phi_k$. Now by Lemma~\ref{TVbound} and the assumption of $\mu$-non degeneracy we conclude that the permutations $\pi_1, \pi_2, \cdots, \pi_j$ are all distinct and that $\phi_1, \dots , \phi_j$ are all at most $1- \epsilon$. For each of these $j$ permutations, we will build an ordered block structure $\mathcal{A}_i$ where the total size of the sets in it is at most $2j$, with the property that $\pi_i$ satisfies $\mathcal{A}_i$ but $\pi_1, \cdots ,\pi_{i-1}, \pi_{i+1}, \cdots ,\pi_j$ do not. We will do this as follows: For each choice $\pi_\ell$ of a permutation in the list $\pi_1, \cdots ,\pi_{i-1}, \pi_{i+1}, \cdots ,\pi_j$ we add to $\mathcal{A}_i$ two consecutive elements of $\pi_i$ whose order is reversed in $\pi_\ell$. This completes the construction. 

Now we will pick an additional $2(k-j)$ elements $x_1,y_1,x_2,y_2, \cdots, x_{k-j},y_{k-j}$ in the following manner. First none of these elements should occur in any of the ordered block structures $\mathcal{A}_i$. Second we want that each pair $(x_i,y_i)$ is consecutive in $\pi_1$. Third, we want the pairs to be as early as possible in $\pi_1$. What is the largest rank that we will need to use to select these additional elements? There are at most $2j^2$ elements among the $j$ ordered block structures and at worst the gap between these elements in $\pi_1$ is one so that we need to use at most rank $2j^2 + 2j^2 + 2(k-j) \leq 4k^2$. 

\begin{proof}[Proof of First Case]
Now consider the set $\mathcal{S}$ of permutations where the first $2(k-j)$ elements are $x_1,y_1, \cdots, x_{k-j},y_{k-j}$ in that order except up to possibly reversing the order of some pairs $x_i$ and $y_i$. There are $2^{k-j}(n-2(k-j))!$ such permutations and we will let $v_i$ denote the vector $v(M_i)$ restricted to the indices corresponding to $\mathcal{S}$. Set $Y = \{ x_1,y_1, \cdots, x_{k-j},y_{k-j} \}$ and $X = [n] \backslash Y$. As in Definition~\ref{def:tensor} we can form a $2 \times 2 \times \cdots \times 2 \times (n-2(k-j))!$ dimensional rank one tensor of order $k-j+1$ whose entries represent the probability of any permutation in $\mathcal{S}$, which using Corollary~\ref{blockStructure} can be written as
$\Pr_{M_i}[\pi \in \mathcal{S}] T_i \otimes v(M_i(\pi_i \vert_{X}))$ where
$$T_i = v(M_i(\pi_i \vert_{\{x_1,y_1 \}})) \otimes v(M_i(\pi_i \vert_{\{x_2,y_2 \}})) \otimes \cdots \otimes v(M_i(\pi_i \vert_{\{x_{k-j},y_{k-j} \}}))$$
Furthermore $\Pr_{M_1}[\pi \in \mathcal{S}] \geq \frac{\phi_1^{8k^3}}{n^{2k}}$ since the probability that the elements $x_1,y_1, \cdots, x_{k-j},y_{k-j}$ are the first $2(k-j)$ elements in that order is at least $\frac{\phi_1^{4k^2(2(k-j))}}{n^{2(k-j)}} \geq \frac{\phi_1^{8k^3}}{n^{2k}}$.

Now for any $1 \leq i \leq k$ and any $1 \leq a \leq k-j$ define the vector $v_{i, a}$ as follows: If $x_a$ occurs before $y_a$ in $\pi_i$ then set $v_{i, a}= (\frac{\phi_i}{1+\phi_i}, \frac{-1}{1+\phi_i})$.  Otherwise set $v_{i, a}= (\frac{1}{1+\phi_i}, \frac{-\phi_i}{1+\phi_i})$.  Note that by construction we have that $v(M(\pi_i \vert_{\{x_a, y_a\}}))$ and $v_{i, a}$ are orthogonal. Now define the tensor
$$T = v_{j+1, 1} \otimes v_{j+2, 2} \otimes \cdots \otimes v_{k, k-j}$$ which has the key property that $\langle T, T_i \rangle = 0$ for $j+1 \leq i \leq k$. Next we lower bound $\langle T, T_1 \rangle$. First we note for any $1 \leq \ell \leq k-j$ we have
$|\langle v(M(\pi_1 \vert_{\{x_1,y_1 \}})), v_{j+\ell, \ell}  \rangle | \geq \frac{|\phi_1 - \phi_{j+\ell}|}{4} \geq \frac{1}{640n^{8k+5}}\frac{\epsilon^{4k^2}}{(k+1)^{2k^2+4k}}$.
Thus we conclude
$|\langle T, T_1 \rangle| \geq \Big( \frac{1}{640n^{8k+5}}\frac{\epsilon^{4k^2}}{(k+1)^{2k^2+4k}} \Big)^k$.
Now putting it all together we have 
$$\|z_1v(M_1) + \cdots + z_k v(M_k)\|_1 \geq \| \sum_{i =1}^{j} z_i \Pr_{M_i}[\pi \in \mathcal{S}] \langle T, T_i \rangle v(M_i(\pi_i \vert_{X}))\|_1$$
We claim that the permutations $\pi_1 \vert_{X}, \cdots, \pi_j \vert_{X}$ must all be distinct because we chose the elements in $Y$ to avoid each ordered block structure $\mathcal{A}_i$. Moreover the scaling parameters are all $\frac{1}{160n^{8k+5}}\frac{\epsilon^{4k^2}}{(k+1)^{2k^2+4k}}$ close to each other. We can apply Lemma \ref{RobustKruskal} to lower bound the right hand side to complete the proof. 
\end{proof}

We remark that  everything in the first case works as is if $\phi_i \geq \frac{\epsilon}{2nk}$ and $|z_i|> \frac{1}{k}$ for any $i$. Thus in the remaining case we can assume that for every $i$ either $\phi_i < \frac{\epsilon}{2nk}$ or $|z_i|\leq  \frac{1}{k}$. 

\begin{proof}[Proof of Second Case]
Assume without loss of generality that $z_1 =1$. Then as shown in the proof of Lemma~\ref{TVbound} we have $\Pr_{M_1}[\pi_1] \geq 1-\frac{\epsilon}{2k}$. For each $j$ with $\phi_j< \frac{\epsilon}{2nk}$, from Lemma~\ref{TVbound} and our choice of $\mu$ we have that $\pi_i \neq \pi_j$ and thus we can find a pair of elements, say $(x_j, y_j)$, that are consecutive in $\pi_i$ and $x_j$ is ranked higher, but in $\pi_j$ they occur in the opposite order. Let $\mathcal{S}$ be the set of permutations where for each such pair $x_j$ is ranked higher than $y_j$. Thus we have that for any $j$, either $\Pr_{M_j}[\pi \in \mathcal{S}] \leq \frac{\epsilon}{2nk}$ or $|z_j| < \frac{1}{k}$. Putting this all together we have
$$\|z_1v(M_1) + \cdots + z_k v(M_k)\|_1 \geq 1-\frac{\epsilon}{2k} - (k-1)\max\Big (\frac{1}{k},\frac{\epsilon}{2nk} \Big) \geq \frac{1}{2k}$$
which completes the proof. 
\end{proof}

Now we are ready to prove our main identifiability result:
\begin{theorem} \label{fullIdentifiability}
Consider two mixtures of Mallows models
$$M = w_1 M(\phi_1, \pi_1) + \cdots + w_i M(\phi_i, \pi_i) \mbox{ and } M' = w'_1 M(\phi'_1, \pi'_1) + \cdots + w'_{i'} M(\phi'_{i'}, \pi_{i'})$$
where $i, i' \leq k$. Suppose that both mixtures are $(\mu, \alpha)$-non degenerate and set $\epsilon = \frac{\mu^2}{10n^3}$.  If for some parameter $\theta$ we have
$$\|\sum_i w_i v(M(\phi_i, \pi_i)) -\sum_{i'} w'_{i'} v(M(\phi_{i'}, \pi_{i'}))\|_1  \leq \big( \frac{\epsilon \theta \alpha}{nk} \big)^{(10k)^{6k}}$$
Then $i = i'$ and there is a matching between the components in the two mixtures so that across the matching the components are $\theta$-close in total variation distance and the mixing weights are also $\theta$-close. 
\end{theorem}
\begin{proof}
We apply Lemma \ref{Identifiability} to conclude that there is a pair of components that is $(\frac{\epsilon \theta \alpha}{nk} )^{(10k)^{6k-4}}$-close
in total variation distance.  From the non-degeneracy assumption, these  two components cannot be from the same mixture. Without loss of generality suppose they are $M_1 = M(\phi_1, \pi_1)$ and $M'_1 = M(\phi'_1, \pi'_1)$. Now if we replace the terms $w_1v(M_1) - w_1'v(M_1')$ in the sum by $(w_1-w_1')v(M_1)$, we increase the $\ell_1$ norm by at most  $2\big( \frac{\epsilon \theta \alpha}{nk} \big)^{(10k)^{6k-4}}$.  We can then repeat the above argument and combine two more terms of the mixture. Neither of these components can involve $M_1$ since $M_1$ cannot be closer than $\frac{\mu}{2}$ in total variation distance to any of the other remaining components. Thus we will combine two new components, which we can assume without loss of generality are $M_2,M_2'$. Also from the non-degeneracy assumption, any component that has not been combined with another must have a mixing weight of at least $\alpha$. Thus $i = i'$ and when all components have been combined we have
$$\|(w_1-w_1')v(M_1) + \cdots + (w_k-w_k')v(M_k)\|_1 \leq \big( \frac{\epsilon \theta \alpha}{nk} \big)^{(10k)^4}$$
 Since we assumed the $M_i$ are all $\mu$-far from each other in total-variation distance, applying Lemma \ref{Identifiability} in the contrapositive ensure that $\max(|w_i - w_i'|) \leq \theta$ which completes the proof. 
\end{proof}

\section{The General Algorithm}

Here we leverage the tools and ideas that we developed in the previous section to give a polynomial time algorithm for learning mixtures of Mallows models that works for any constant number of components. We have already seen the key ingredient \---- test functions that isolate a single component from the rest of the mixture. In the context of polynomial identifiability, we knew the parameters of the mixture which we used to construct small ordered block structures that allow us to focus on parts of the distribution that have a convenient analytic form, but still capture the differences between the base permutations. Here we will use variants of the same type of arguments, but where we guess the relevant portions of the ordered block structures from which we follow the same recipe to construct test functions. If our guess is correct, we will succeed in learning the base permutation of some component. In our setting there will be a constant number of guesses, so we will be able to construct a list of candidate mixtures at least one of which is close to the true mixture. We can then appeal to our identifiability results to test find a mixture that is close on a component-wise basis. 

Because we will need to handle components with small scaling parameters separately, it will be more convenient to work with vectorizations of the low order moments of a distribution than the distribution itself. 

\begin{definition}
For a Mallows model $M$ on $n$ elements, let $v_c(M)$ denote the vectorization of the order $c$ moments of $M$. In particular $v_c(M)$ has $\binom{n}{c}n(n-1) \cdots (n-c+1)$ entries and we interpret each coordinate as a choice of a subset $S \subset [n]$ of size $c$ and a placement of its elements. The value of the entry in $v_c(M)$ is the probability under $M$ that the elements in $S$ are placed in the corresponding locations. 
\end{definition}

Note that the sum of the entries in $v_c(M)$ is larger than one, because the values of its entries are the probabilities of events that are (usually) not disjoint. We remark that all of the proofs of polynomial identifiability, where we prove lower bounds on the $\ell_1$-norm of linear combinations of vectorizations of Mallows models, carry over to the case when we use $v_c(M)$ instead provided that $c \geq 10k^2$. This follows by observing that all of the events we used can be defined in terms of the placement of at most $c$ elements and the probabilities of these events can be computed by adding up an appropriate set of the entries of $v_c(M)$ (corresponding to disjoint events themselves) instead of $v(M)$. Thus we have the following corollary:

\begin{corollary}
\label{localIdentifiability}
Consider two mixtures of Mallows models
$$M = w_1 M(\phi_1, \pi_1) + \cdots + w_i M(\phi_i, \pi_i) \mbox{ and } M' = w'_1 M(\phi'_1, \pi'_1) + \cdots + w'_{i'} M(\phi'_{i'}, \pi_{i'})$$
where $i, i' \leq k$. Suppose that both mixtures are $(\mu, \alpha)$-non degenerate and set $\epsilon = \frac{\mu^2}{10n^3}$.  If for some parameter $\theta$ we have
$$\|\sum_i w_i v_c(M(\phi_i, \pi_i)) -\sum_{i'} w'_{i'} v_c(M(\phi_{i'}, \pi_{i'}))\|_1  \leq \big( \frac{\epsilon \theta \alpha}{nk} \big)^{(10k)^{6k}}$$
Then $i = i'$ and there is a matching between the components in the two mixtures so that across the matching the components are $\theta$-close in total variation distance and the mixing weights are also $\theta$-close. \end{corollary}

Next we give an outline of our algorithm: In Claim~\ref{removal} we give an algorithm for finding and removing components with small scaling parameter. The intuition is that such components often generate their own base permutation, so if we take a small number of samples and find all the permutations that occur somewhat frequently, we will have a superset of the base permutations of components with small scaling parameter. We then remove their contribution to the order $c$ moments to generate a list of candidate vectors, at least one of which is close to the true order $c$ moments of the submixture of components without small scaling parameter. In Corollary~\ref{learn2} we give an algorithm that mimics the proof of Lemma~\ref{Kruskal} and Lemma~\ref{Identifiability} but wherever the construction of a test function to isolate a component uses knowledge of the mixture, we guess. The algorithm outputs a list of candidate parameters with the property that accurate estimates of each component in the true mixture appear on the list. We then consider all $k$ tuples of components to form a list of candidate mixtures. Finally in Corollary~\ref{corr:test} we give an algorithm for testing whether a candidate mixture is close to the true mixture. The intuition is we can generate our own samples from a candidate mixture to compute its lower order moments and check whether these are close to the lower order moments of the true mixture. Corollary~\ref{localIdentifiability} tells us that if this check passes then the candidate mixture is indeed component-wise close to the true mixture.

\subsection{Finding Components with Small Scaling Parameters}

When the scaling parameter of a Mallows model is small enough, it generates its own base permutation the majority of the time. Using this intuition, we show that we can take few samples and guess the base permutations of all the components with small scaling parameter and then essentially remove them from the mixture. First we show how to generate a list of few candidates and make precise how small we need the scaling parameter to be:

\begin{fact}\label{fact:list}
Consider a mixture of Mallows models
$M = w_1 M(\phi_1, \pi_1) + \cdots + w_i M(\phi_i, \pi_i) $
where all mixing weights are at least $\alpha$. There is an algorithm that takes $m = \frac{10k\log \frac{1}{\delta}}{\alpha^2}$ samples from $M$ and runs in polynomial 
time and outputs a list $L$ of permutations of size at most $\frac{4}{\alpha}$ so that with probability at least $1-\frac{\delta}{3}$ all $\pi_i$ for which $\phi_i <\frac{1}{2n}$ are included in the list.
\end{fact}
\begin{proof}
Suppose $\phi_i < \frac{1}{2n}$. Then the probability that a draw from $M(\phi_i, \pi_i)$ is $\pi_i$ is at least $$w_i(1-\phi_i)^n \geq \frac{w_i}{2} \geq \frac{\alpha}{2}$$ Now we take $m$ samples from $M$ and add to the list all permutations that appear at least $\frac{\alpha}{4}$ fraction of the time.  There are at most $\frac{4}{\alpha}$ such permutations and with $1- \frac{\delta}{3}$ probability, every $\pi_i$ with $\phi_i \leq \frac{1}{2n}$ is added to the list. 
\end{proof}

Now we show how to essentially remove these components from the mixture. More precisely, we guess which permutations on the list are the base permutations of components with small scaling parameter, then grid search over their mixing weights and scaling parameters. We use these estimates to remove the contribution of the components with small scaling parameter from the order $c$ statistics in $v_c(M)$. 

\begin{claim}\label{removal}
Consider a mixture of Mallows models
$M = w_1 M(\phi_1, \pi_1) + \cdots + w_i M(\phi_i, \pi_i) $
where all mixing weights are at least $\alpha$. Suppose that $\phi_1, \dots , \phi_j<\frac{1}{2n}$ and $\phi_{j+1}, \dots , \phi_k \geq \frac{1}{2n}$. Let $c=10k^2$. There is an algorithm that takes $$m = \Big( \frac{nk\log \frac{1}{\delta}}{\epsilon \theta \alpha} \Big)^{(10k)^{8k}}$$
samples and runs in polynomial time and outputs a list of vectors $v'$ of polynomial size so that with probability at least $1- \frac{\delta}{3}$ at least one $v'$ satisfies
$$\|v' - \sum_{a = j+1}^k w_a v_c(M_a) \|_1 \leq \frac{1}{k}\bigg( \frac{\epsilon\alpha}{2nk} \bigg)^{60k^4}$$

\end{claim}
\begin{proof}
First run the algorithm in Fact~\ref{fact:list} to generate a candidate list of base permutations. Next from samples from $m$, we can form an empirical estimate $v$ for $v_c(M)$ which satisfies
$$\|v - v_c(M)\|_1 \leq \frac{1}{k^2}\bigg(\frac{\epsilon\alpha}{2nk} \bigg)^{60k^4}$$  Finally set 
$\gamma = \big( \frac{\epsilon \theta \alpha}{nk \log \frac{1}{\delta}} \big)^{(10k)^{10k}}$
and search over the $\gamma$-grid of all possible mixing weights and scaling parameters for each permutation on the list. For each possibility $w'_1 M(\phi'_1, \pi'_1), \cdots, w'_j M(\phi'_j, \pi'_j)$ we can estimate their contribution 
$$v' = \sum_{a = 1}^j w'_a v_c(M(\phi'_a, \pi'_a))$$
The key point is if $|w_a - w_a'| \leq \gamma, |\phi_a - \phi_a'| \leq \gamma, \pi_a' = \pi_a$ for $1 \leq a \leq j$ then using Lemma~\ref{TVbound} we have
$$\|w'_a v_c(M(\phi'_a, \pi'_a)) - w_a v_c(M(\phi_a, \pi_a))\|_1 \leq \binom{n}{c}\Big ( (10n^3|\phi_a - \phi_a'|)^{1/2}+|w_a - w_a'| \Big) \leq  n^c(10n^3\gamma)^{1/2} $$
and consequently
$$\|v' - \sum_{a = j+1}^k w_a v_c(M_a) \|_1 \leq  \|v' - \sum_{a=1}^j w_a v_c(M(\phi_a, \pi_a))\|_1  + \|v - v_c(M)\|_1 \leq \frac{1}{k}\bigg( \frac{\epsilon\alpha}{2nk} \bigg)^{60k^4}$$
which completes the proof. 
\end{proof}

\subsection{Finding a Single Component}

In this subsection, we focus on the problem of recovering a single component from a mixture of Mallows models. Our algorithms will closely parallel the polynomial identifiability results in Lemma~\ref{Kruskal} and Lemma~\ref{Identifiability}. More precisely, we will modify the test functions we used in those results to turn them into algorithms for isolating a single component, first focusing on the case when all the scaling parameters are the same and then the general setting where they can be different.

\begin{lemma}\label{SpecialLearn}
Consider a collection of $k$ Mallows models $M(\phi, \pi_1), \cdots, M(\phi, \pi_k)$ with distinct base permutations and where $\phi \leq 1 - \epsilon$. Suppose $\phi$ is known and we are given a vector $v$ with
$$\|v - \sum_i z_i v_c(M(\phi, \pi_i)) \|_1 \leq \Big(\frac{\epsilon^{k}}{nk}\Big)^{4k}$$
and a constant $c \geq 10k^2$. Finally suppose $z_1 \geq 1$. There is a polynomial time algorithm to output a list of at most $n^{4k}((2k)!)^k$ permutations that contains $\pi_1$.  
\end{lemma}
\begin{proof} 
In the proof of Lemma~\ref{Kruskal} we showed that if the base permutations are known, there is a procedure for constructing an ordered block structure $\mathcal{B} = S_1, \cdots, S_j$ containing at most $2k$ elements so that $\pi_1$ satisfies $\mathcal{B}$ and any other $\pi_j$ does not satisfy it even as a order structure. (In the proof of Lemma~\ref{Kruskal} we treated $\mathcal{B}$ as a block structure but here we will use it slightly differently.) Here we guess such a $\mathcal{B}$ along with the relative orderings of each of $S_1, \cdots, S_j$ in each of $\pi_2, \cdots, \pi_k$. We will give an algorithm which, if each of these guesses are correct, outputs $\pi_1$. 

First we choose two more elements $x$ and $y$ not in $\mathcal{B}$. Now define the $|S_1|! \times \cdots \times |S_j|!$ tensor $T_{i, 1}$. Each entry corresponds to an ordering of $S_1, \cdots, S_j$ and in it we put the probability that a ranking drawn from $M(\phi, \pi_1)$ has these relative orderings, satisfies $\mathcal{B}$ and also $x$ is ranked higher than $y$. Similarly we define $T_{i, 2}$ except that $y$ is ranked higher than $x$. Note that the entries of $T_{i, 1}$ and $T_{i, 2}$ can be efficiently constructed from $v_c(M(\phi, \pi_1))$ because they only depend on the joint distribution of a set of elements of size at most $c$.
Let $p_{i,1} = \Pr_{M_i}[\pi \in \mathcal{S}_{\mathcal{B}} \mbox{ and } \pi(x) < \pi(y)]$. Then from Corollary~\ref{blockStructure} we have that
$$T_{i, 1} = p_{i, 1} v(M(\phi, \pi_i \vert_{S_1})) \otimes \cdots \otimes v(M(\phi, \pi_i \vert_{S_j}))
$$
and similarly for $T_{i, 2}$ where $p_{i, 2}$ is defined analogously, but where $y$ is higher ranked than $x$. From Lemma~\ref{fix}, we have that $p_{1, 1}+p_{1, 2} \geq \frac{1}{n^{4k}}$.

Because we have assumed we guessed the relative ordering of each $S_a$ in each $\pi_i$ correctly, we can repeat the construction in Lemma~\ref{Kruskal} to construct unit vectors $v_1, \cdots, v_j$ that satisfy
\begin{enumerate}

\item[(1)] $\langle v_a, v(M(\phi, \pi_i \vert_{S_a})) \rangle = 0$ whenever $\pi_i \vert_{S_a} \neq \pi_1 \vert_{S_a} $ and

\item[(2)] $\langle v_a, v(M(\phi, \pi_1 \vert_{S_a})) \rangle \geq \frac{1}{|S_a|!}(\frac{\epsilon^{|S_a|}}{\sqrt{|S_a|!}})^k$ for all $a$
\end{enumerate}
Now the same way $T_{i, 1}$ and $T_{i, 2}$ can be efficiently constructed from $v_c(M(\phi, \pi_1))$, we can construct $T_1$ and $T_2$ from the estimate $v$. It follows that
$$\|v - \sum_i z_i v_c(M(\phi, \pi_i)) \|_1 \geq \|T_1 - \sum_i z_i T_{i, 1}\|_1  \geq | \langle v_1 \otimes \cdots \otimes v_j, T_1 \rangle - z_1 \langle v_1 \otimes \cdots \otimes v_j, T_{1, 1} \rangle | $$
and similarly for $| \langle v_1 \otimes \cdots \otimes v_j, T_2 \rangle - z_1 \langle v_1 \otimes \cdots \otimes v_j, T_{1, 2} \rangle | $.

Now if $x$ is higher ranked than $y$ in $\pi_1$, we have $p_{1, 2} \leq \phi p_{1, 1}$ because interchanging the positions of $x$ and $y$ reduces the number of inversions relative to $\pi_1$ by at least one. Thus $p_{1,1} - p_{1, 2} > \frac{\epsilon}{2n^{4k}}$ from which we conclude
$$z_1 \langle v_1 \otimes \cdots \otimes v_j, T_{1, 1} \rangle - z_1 \langle v_1 \otimes \cdots \otimes v_j, T_{1, 2} \rangle  \geq (p_{1,1}-p_{1,2})\prod_{a=1}^{j} \frac{1}{|S_a|!}\Big(\frac{\epsilon^{|S_a|}}{\sqrt{|S_a|!}}\Big)^k > \Big(\frac{\epsilon^{k}}{nk}\Big)^{4k}$$
And if $y$ is ranked higher then the inequality above holds with $T_{1, 1}$ and $T_{1,2}$ as well as $p_{1, 1}$ and $p_{1,2}$ interchanged. Thus we can deduce that $x$ is ranked higher than $y$ in $\pi_1$ if $$\langle v_1 \otimes \cdots \otimes v_j, T_1 \rangle > \langle v_1 \otimes \cdots \otimes v_j, T_2 \rangle$$ and if not then $y$ is ranked higher.  We can then repeat this for all pairs of elements $x$ and $y$ to recover the relative ordering of all elements outside of $\mathcal{B}$ and then we can guess all possible positions of the at most $2k$ elements in $\mathcal{B}$, which completes the proof. 
\end{proof}

Next we give an algorithm for isolating a single component when the scaling parameters are allowed to be different. In addition to the usual assumptions, we will assume that each scaling parameter is at least $\frac{1}{2n}$, since the algorithm in Claim~\ref{removal} allows us to remove such components from our estimates of the order $c$ moments. 

\begin{lemma}\label{learn1}
Consider a mixture of $k$ Mallows models
$M = w_1 M(\phi_1, \pi_1) + \cdots + w_k M(\phi_k, \pi_k) $ where for each $i$, $\alpha \leq w_i$ and the total variation distance between any two components is at least $\mu$. Furthermore suppose that for all $i$, $\frac{1}{2n} < \phi_i < 1- \epsilon$ where $\epsilon = \frac{\mu^2}{10n^3}$. Let $c=10k^2$ and $\theta$ be the target accuracy. Suppose we are given a vector $v$ with 
$$\|v- \sum_i w_i v_c(M(\phi_i, \pi_i))\|_1 \leq \Big( \frac{\epsilon\alpha}{2nk} \Big)^{60k^4}$$  
There is a polynomial time algorithm to output a list of candidate parameters $(w, \phi, \pi)$ so that for some $\ell$ there is at one entry in the list that is $\theta$-close \---- i.e. it satisfies $|w-w_\ell | \leq \theta, |\phi-\phi_\ell| \leq \theta$ and $ \pi = \pi_\ell$. 
\end{lemma}

\begin{proof} 
Throughout this proof set $\beta = \min(\theta, ( \frac{\epsilon\alpha}{2nk} )^{150k^4})$. First we search over the $\beta$-grid of all possible mixing weights and scaling parameters. We will give an algorithm that, if all our guesses are correct in the sense that all these parameters are within $\beta$ of their true values, outputs a list of candidate parameters that contains an entry that is $\theta$-close to at least one component. Suppose that our guesses are $\phi_1', \cdots , \phi_k'$ and $w_1', \cdots w_k'$. 

Now consider all $\phi_i'$ such that $|\phi_i'-\phi_1'| > (\frac{(\epsilon\alpha)^{k}}{nk})^{40k^2}$ and suppose without loss of generality that these are $\phi_{j+1}', \cdots, \phi_k'$.
In the proof of Lemma~\ref{Identifiability} we showed that based on the guesses $\phi_1', \cdots , \phi_k'$ and if the base permutations are known, for each $j+1 \leq i \leq k$ it is possible to find two elements $(x_{i-j}, y_{i-j})$ that are consecutive in $\pi_1$ and occur in the opposite order in $\pi_i$. Moreover all of these elements have rank at most $4k^2$. Finally if we set $X = [n] \backslash \{ x_1,y_1, \cdots , x_{k-j}, y_{k-j} \}$, then $\pi_1 \vert_X, \cdots , \pi_j \vert_X$ are all distinct (we accomplished this by building a collection of ordered block structures, but here we do not explicitly need them). Here we guess all of these pairs of elements for each $j+1 \leq i \leq k$. 

We will follow the proof of Lemma~\ref{Identifiability} to construct test functions that can isolate a single component. Let $M'_i = M(\phi'_i, \pi_i)$. If we were given $v_c(M'_i)$ we could construct a tensor of the form
$$S_i = \Pr_{M'_i}[x \in \mathcal{S}] T_i \otimes v_{c-2(k-j)}(M(\phi'_i, \pi_i \vert_X))$$
where we set $T_i  = v(M(\phi'_i, \pi_i \vert_{\{x_1,y_1 \}})) \otimes \cdots \otimes v(M(\phi'_i, \pi_i \vert_{\{x_{k-j},y_{k-j} \}}))$ and $\mathcal{S}$ is the set of permutations where the first $2(k-j)$ elements are $x_1,y_1, \cdots, x_{k-j},y_{k-j}$ in that order except up to possibly reversing the order of some pairs $x_i$ and $y_i$. Now we can repeat the construction in Lemma~\ref{Identifiability} to find vectors $v'_{i,a}$ (using $\phi_1', \cdots , \phi_k'$ instead of $\phi_1, \cdots , \phi_k$ e.g. so that now $v(M(\phi_i', \pi_i \vert_{\{x_a, y_a\}}))$ and $v'_{i, a}$ are orthogonal). Now define the tensor
$$T = v'_{j+1, 1} \otimes v'_{j+2, 2} \otimes \cdots \otimes v'_{k, k-j}$$

With this construction in hand, we can compute
$$\langle T , \sum_i w'_i S_i \rangle = \sum_{i =1}^j z_i v_{c-2(k-j)}(M(\phi'_i, \pi_i \vert_X)) \mbox{ where } z_i = \Pr_{M'_i}[x \in \mathcal{S}] w'_i \langle T, T_i \rangle $$
The idea now is to use the algorithm in Lemma~\ref{SpecialLearn} since each of the remaining components have scaling parameters that are close. To accomplish this, we need to lower bound $z_1$ and upper bound the error in estimating the sum above using the vector $v$ we are given and the error introduced by making all the scaling parameters the same. First we claim that
$$|z_1| \geq \Big( \frac{\epsilon\alpha}{2nk} \Big)^{50k^4}$$
which follows from the assumption we made that $x_1,y_1, \cdots, x_{k-j},y_{k-j}$ have rank at most $4k^2$ and the resulting bound $\Pr_{M'_1}[x \in \mathcal{S}] \geq \frac{\phi_1'^{8k^3}}{n^{2k}}$ along with the bound 
$$|\langle v'_{i, a}, v(M(\phi'_1, \pi_1 \vert_{\{x_a, y_a\}} \rangle | \geq \frac{|\phi_i'-\phi_1'|}{4} \geq \frac{1}{4} \Big( \frac{(\epsilon\alpha)^{k}}{nk} \Big) ^{40k^2}$$
Second, we guess the index $\ell$ where $z_\ell$ is the largest. We will bound the error in replacing each $\phi'_i$ with $\phi'_\ell$. Using the fact that the sum of the entries in $v_{c-2(k-j)}(M(\phi, \pi)) \leq n^c$ and any Mallows model $M(\phi, \pi)$ along with Lemma~\ref{TVbound} we have that
$$\|v_{c-2(k-j)}(M(\phi'_\ell, \pi_i \vert_X )) - v_{c - 2(k-j)}(M(\phi'_i, \pi_i \vert_X))\|_1 \leq
 2n^c\Big( 10n^3 |\phi'_\ell - \phi'_i| \Big)^{1/2} \leq \Big( \frac{(\epsilon\alpha)^k}{nk} \Big) ^{8k^2}$$
Finally from the given $v$ we can construct a tensor $T_v$ in the same manner that we constructed $T_i$ from $v_c(M'_i)$ above, which gives:
$$ \|T_v - \sum_{i =1}^j z_i v_{c-2(k-j)}(M(\phi'_\ell, \pi_i \vert_X))\|_1 \leq \|T_v - \sum_{i =1}^j z_i v_{c-2(k-j)}(M(\phi'_i, \pi_i \vert_X))\|_1 + jz_\ell \Big( \frac{(\epsilon\alpha)^k}{nk} \Big)^{8k^2}$$

To upper bound the right hand side in the above expression, we can use the assumption about $v$ in the statement of the lemma.  We can then invoke Lemma ~\ref{TVbound} and the assumption that $|\phi_i - \phi_i'| \leq \beta$ to replace $\phi_i$ with $\phi_i'$.  With this upper bound, we apply the algorithm in Lemma~\ref{SpecialLearn} to learn the permutation $\pi_\ell \vert_X$. We then output the candidate parameters $w'_\ell, \phi'_\ell$ and $\pi'_\ell$ (where we guess the remaining positions of the remaining $2(k-j)$ elements in $[n] \setminus X$ to fill in the rest of $\pi'_\ell$). It is easy to see that when we replace the guesses we made with brute force search, the algorithm still runs in polynomial time. 
\end{proof}

\subsection{Finding the Rest of the Components} 
 It is now straightforward to repeatedly use the algorithm in Lemma~\ref{learn1} to learn and peel off components one by one, which is captured by the following corollary:

\begin{corollary}\label{learn2}
Under the same conditions as Lemma~\ref{learn1}, suppose we are given a vector $v$ with
$$\|v- \sum_i w_i v_c(M(\phi_i, \pi_i))\|_1 \leq \Big( \frac{\epsilon\alpha}{2nk} \Big)^{60k^4}$$ 
There is a polynomial time algorithm to output a list of candidate parameters $(w, \phi, \pi)$ so that for each $\ell$ there is at one entry in the list that is $\theta$-close \---- i.e. it satisfies $|w-w_\ell | \leq \theta, |\phi-\phi_\ell| \leq \theta$ and $ \pi = \pi_\ell$.
\end{corollary}
\begin{proof}
We invoke the algorithm in Lemma~\ref{learn1} and set the target accuracy so that the entry $(w, \phi, \pi)$ that is close to some component on the list that it returns satisfies
$$\|w v_c(M(\phi, \pi)) - w_\ell v_c(M(\phi_\ell, \pi_\ell))\|_1 \leq \frac{1}{k}\Big( \frac{\epsilon\alpha}{2nk} \Big)^{60k^4}$$
More precisely, the algorithm in Lemma~\ref{learn1} guarantees that $\pi = \pi_\ell$ and the mixing weight and scaling parameter are $\theta'$-close. We make $\theta'$ small enough that using the bound $$\|v_c(M(\phi, \pi)) - v_c(M(\phi_\ell, \pi_\ell))\|_1 \leq 2n^c d_{TV}(M(\phi, \pi), M(\phi_\ell, \pi_\ell))$$ along with Lemma~\ref{TVbound} yields the desired inequality. 
Now for each candidate parameters $(w, \phi, \pi)$ on the list, compute 
$$u = v - w v_c(M(\phi, \pi))$$
and repeat in this manner $k-1$ more times. It is easy to see that the $\ell_1$-norm of the difference between $u$ and the components in the mixture that have not been found (on some list) yet will not exceed $( \frac{\epsilon\alpha}{2nk})^{60k^4}$.
\end{proof}

Of course, once we have a list where an accurate estimate of every component in the mixture appears, we can try all possible $k$ tuples of parameters on the list in order to generate a list of candidate mixtures, at least one of which is component-wise close to the true mixture. All that remains is to hypothesis test all of these possibilities. This is slightly more complicated in our setting, because we want to select a candidate that is not just close in total variation distance as a mixture, but even in a component-wise sense. 

\subsection{Testing Component-wise Closeness}

Here we show how to test whether a pair of mixtures of Mallows models are component-wise close. This essentially follows by invoking Corollary~\ref{localIdentifiability} and standard arguments. 

\begin{corollary}\label{corr:test}
Suppose we are given sample access to a mixture of $k$ Mallows models $M = w_1 M(\phi_1, \pi_1) + \cdots + w_k M(\phi_k, \pi_k) $ and an estimate $M' = w'_1 M(\phi'_1, \pi'_1) + \cdots + w_k M(\phi'_k, \pi'_k) $ on $n \geq 10k^2$ elements where both mixtures are $(\mu, \alpha)$-non degenerate.  Se t $\epsilon = \frac{\mu^2}{10n^3}$. There is an algorithm which given 
$$m = \Big( \frac{nk \log \frac{1}{\delta'}}{\epsilon \theta \alpha} \Big)^{(10k)^{8k}}$$
samples from $M$ runs in polynomial time and succeeds in accepting when the mixtures are component-wise $\gamma$-close for
$$\gamma = \Big( \frac{\epsilon \theta \alpha}{nk \log \frac{1}{\delta'}} \Big)^{(10k)^{10k}}$$ and rejects when they are component-wise $\theta$-far and succeeds with probability at least $1 - \delta'$. 
\end{corollary}

\begin{proof} As usual, set $c=10k^2$. From the choice of $\gamma$ we have that if the mixtures are component-wise $\gamma$-close then
$$\| \sum_i w_i v_c(M(\phi_i, \pi_i)) - \sum_i w'_i v_c(M(\phi'_i, \pi'_i))\|_1 \leq \Big( \frac{\epsilon \theta \alpha}{nk} \Big)^{(10k)^{8k}}$$
And from Corollary~\ref{localIdentifiability} we know a weak converse that if the above bound (with, say, an extra factor of four) holds then the mixtures must be component-wise $\theta$-close. Now we can estimate the above quantity using $m$ samples from $M$ and by generating $m$ samples from $M'$. The latter can be done efficiently through any of the known sampling schemes for Mallows models, e.g. \cite{Insertion}. We accept if the above bound holds (with an extra factor of two) and reject otherwise, and by standard concentration argument it is easy to see that the failure probability is at most $\delta'$. 
\end{proof}

We are now ready to prove Theorem \ref{generalalgorithm}:

\begin{proof} 
First we run the algorithm from Claim~\ref{removal}. Then we run the algorithm in Corollary~\ref{learn2} to output a list of candidate parameters, and we consider all $k$ tuples of components to form a list of candidate mixtures. On this list, we eliminate candidate mixtures in which $(1)$ a pair of components has the same base permutation and the scaling parameters are within $\frac{\epsilon}{10}$ and $(2)$ any component has a scaling parameter more than $1-\frac{\epsilon}{2}$ or mixing weight less than $\frac{\alpha}{2}$. For $$\gamma = \Big( \frac{\epsilon \theta \alpha}{nk \log \frac{1}{\delta'}} \Big)^{(10k)^{10k}}$$
the candidate mixture that is component-wise $\gamma$-close to the true mixture will not be eliminated. We claim all the remaining candidates have components that are at least $\frac{\epsilon}{40}$-far in total-variation distance. To see this, for any two components with different base permutations, we use Claim \ref{simple}.  Otherwise, pick two elements that are consecutive in the (shared) base permutation.  Say the scaling parameters are $\phi_i,\phi_j$.  The probability that the two elements occur in order in a sample from the first Mallows model is $\frac{1}{1+\phi_i}$ while for the second Mallows model the probability is $\frac{1}{1+\phi_j}$.  The difference of these quantities is at least $\frac{|\phi_i - \phi_j|}{4}$ giving us the lower bound on the total variation distance.  

Now for each remaining candidate we apply the algorithm in Corollary~\ref{corr:test} and setting $\mu = \frac{\epsilon}{40}$ and $\delta'$ sufficiently small so that the probability of failing on any candidate in our polynomially sized list is at most $\frac{\delta}{3}$. We know that if the algorithm succeeds at deciding which mixtures are component-wise close and far, that we will accept on at least one of the candidate mixtures (the one that is $\gamma$-close to the true mixture) and that anything we accept will be component-wise $\theta$-close.  Finally, the overall failure probability is at most $\delta$ and the number of samples used is at most $m = ( \frac{nk\log \frac{1}{\delta}}{\epsilon \theta \alpha} )^{(10k)^{10k}}$
which completes the proof. 
\end{proof}

\section{Lower Bounds}

In this section, we show various lower bounds for learning mixtures of $k$ Mallows models. First, we give information-theoretic lower bounds on the sample complexity that show that any algorithm for learning the components (and not merely learning a mixture that is close as a distribution) must take at least an exponential in $k$ number of samples. Second, we define a local model where an algorithm is only allowed to make statistical queries on up to $c$ elements. This framework captures all of our algorithms along with many other natural strategies. We show that there are mixtures of $k$ Mallows models that are far as distributions, but require $n^{\log k}$ queries to distinguish. An interesting open question is to prove a statistical query lower bound, which necessitates finding a much larger collection of mixtures that are hard to distinguish from each other.

\subsection{Sample Complexity Lower Bounds}
In this subsection, we construct two mixtures of $k$ Mallows models that are $(\mu, \alpha)$-non degenerate that are close as mixtures but not on a component-by-component basis. Recall that non-degeneracy avoids more trivial reasons why the components are not learnable from few samples (such as not getting a sample from each component, or not getting a sample from the difference between two components). The following lower bound shows that any algorithm for learning the components a $(\mu, \alpha)$-non degenerate mixture of $k$ Mallows models must take at least $\frac{1}{\mu^k}$ samples:

\begin{lemma}\label{lowerBound1}
For any $\mu \leq \frac{1}{40k^2}$ and $n \geq 40k^2$ there are two mixture of at most $k$ Mallows models $M$ and $M'$
with the following properties:
\begin{enumerate}

\item[(1)] Each mixture is $(\mu, \frac{1}{10 \cdot 2^{2k}})$-non degenerate

\item[(2)] $d_{TV}(M, M') \leq 4(8\mu k)^{2k-1}$

\item[(3)] $M$ and $M'$ are not component-wise $\mu$-close

\end{enumerate}
\end{lemma}

\begin{proof} 
Our strategy is to find a linear combination of vectorizations of Mallows models whose coefficients are lower bounded, but whose $\ell_1$-norm can be much smaller. Then we will interpret this linear combination as a pair of mixtures of Mallows models. Let $r = 2k$ and let $\pi = (1, 2, \cdots, n)$. Now consider $M(\phi_1, \pi), \cdots ,M(\phi_r,\pi)$ where $\phi_1 = \lambda, \phi_2 = 2\lambda, \cdots, \phi_r = r\lambda$ for some $\lambda$ that will be chosen later. Finally let $Z_n(\phi) = (1+\phi) \cdots (1+\phi+ \cdots + \phi^{n-1})$.  Now consider the linear combination
$$v = \sum_{i = 0}^r (-1)^{i}\binom{r-1}{i}Z(\phi_i)v(M(\phi_i, \pi))$$

First, we will bound the $\ell_1$-norm of $v$. Then we will interpret it as the difference of two mixtures and prove that its components are far apart in total variation distance and that its mixing weights are not too small. 

\begin{claim}
$\|v\|_1 \leq \frac{(2nr\lambda)^{r-1}}{1-2nr\lambda}$
\end{claim}
\begin{proof}
Consider a permutation $\pi'$ that has $i$ inversions. Then the value of $v$ in the entry indexed by $\pi'$ is 
$$\lambda^i \sum_{j=0}^{r-1} (-1)^j\binom{r-1}{j} (j+1)^i$$
which is zero when $i \leq r-2$ and is at most $(2r \lambda)^i$ when $i \geq r-1$. We claim that there are at most $n^i$ permutations with exactly $i$ inversions, which can easily be seen by induction by generating the permutation from $i$ swaps of adjacent elements. Together this implies
$$\|v\|_1 \leq \sum_{i=r-1}^{\infty} (2nr\lambda)^i \leq \frac{(2nr\lambda)^{r-1}}{1-2nr\lambda}$$
which completes the proof. 
\end{proof}

Now set $\lambda = \frac{2\mu}{n}$. We can compute
$$\Pr_{M(\phi_i, \pi)}[\pi] = \frac{1}{(1+i\lambda) \cdots (1+ i\lambda + \cdots + (i\lambda)^{n-1})} = \frac{(1-i\lambda)^{n-1}}{(1-(i\lambda)^2) \cdots (1-(i\lambda)^{n})}$$
Using this expression, we can upper and lower bound the probability that the $i$th component generates $\pi$:
\begin{claim}
$1-2i\mu \leq \Pr_{M(\phi_i, \pi)}[\pi] \leq 1-(2i-1)\mu$
\end{claim}
\begin{proof}
First, the numerator can be upper bounded as
$$\Big(1-\frac{2i\mu}{n}\Big)^{n-1} \leq e^{-2i\mu} \leq 1-(2i-0.1)\mu$$
and lower bounded as
$$\Big(1-\frac{2i\mu}{n}\Big)^{n-1} \geq 1-2i\mu$$
The denominator is at most one and at least $1-\frac{(i\lambda)^2}{1-i\lambda}$, which using the lower bound on $n$ gives us the desired bounds. 
\end{proof}

The upper and lower bounds on $\Pr_{M(\phi_i, \pi)}[\pi]$ allow us to lower bound the total variation distance between two components as
$$d_{TV}(M(\phi_i, \pi), M(\phi_j, \pi)[\pi]) \geq |\Pr_{M(\phi_i, \pi)}[\pi] - \Pr_{M(\phi_j, \pi)}[\pi]| \geq \mu$$
and similarly for the total variation distance to the uniform distribution on permutations. 

Now we will interpret $v$ as the difference between two mixtures of at most $k$ Mallows models. If the sum of the entries of $v$ is non-zero, we can increase the coefficient of either $v(M(\phi_1, \pi)$ or $v(M(\phi_2, \pi)$ to make it zero. This increases the $\ell_1$-norm by at most a factor of two. Now we can take the positive terms, renormalize so that the sum of their entries is one, and interpret them as components of a mixture of $k$ Mallows models. We can do the same for the negative terms. It is easy to see that after renormalization, the minimum mixing weight is at least $\frac{1}{10\cdot 2^r}$ using the fact that $Z(\phi_i) \leq \frac{1}{(1-\phi_i)^n} \leq 10$ for all $1 \leq i \leq r$, which completes the proof. 
\end{proof}

\subsection{Lower Bounds Against Local Query Algorithms}
We now introduce a restricted model of learning, that is natural in the context of distributions over permutations, for which we can prove an $n^{\log k}$ lower bound for learning mixtures of $k$ Mallows models.  
\begin{definition}
In the \textit{local query model}, the learner queries a subset of elements $x_1,x_2, \cdots, x_c$ and locations $i_1,i_2, \cdots, i_c$ with a tolerance parameter $\tau$ and is answered with the probability, up to an additive $\tau$, that a sample from the mixture has $x_j$ in position $i_j$ for all $1 \leq j \leq c$. The cost of the query is $\frac{1}{\tau^2}$ and the total cost of an algorithm is the sum of its query costs. 
\end{definition}

We now show the following lower bound:

\begin{theorem} \label{SQL}
Suppose $k = 2^{\ell -1}$ and $2\ell$ divides $n$. Then there is a randomized construction of a pair of mixtures of $k$ Mallows models $M$ and $M'$ with the following properties:
\begin{enumerate}

\item[(1)] $M$ and $M'$ are $(\frac{1}{40},\frac{1}{k})$-non degenerate

\item[(2)] with probability $1/2$, $M = M'$ and otherwise $M,M'$ have the property that every component of $M$ is at least $\frac{1}{40}$-far in total variation distance from every component of $M'$.

\end{enumerate}
Yet any algorithm in the local query model that decides whether $M = M'$ with success probability at least $\frac{2}{3}$ must incur cost at least $\big( \frac{n}{2k} \big)^{\log_2k}$. 
\end{theorem}

We will construct $M$ and $M'$ whose components are far in total variation distance and yet for any local query an algorithm makes, if it has low cost, the query can be answered the same for both $M$ and $M'$. The theorem then follows by flipping a fair coin to decide whether to output the pair $(M, M)$ or the pair $(M, M')$. 
First, we describe the construction. Set the scaling parameter $\phi = 1- (\frac{k}{n})^{1/2}$  for all of the components we will use. Start from the identity permutation $\pi = (1, 2, \cdots, n)$ and divide the $n$ elements into $\ell$ blocks of $\frac{n}{\ell}$ consecutive elements. In each block,  create $\frac{n}{2\ell}$ pairs of consecutive elements. There are $2^{\ell}$ permutations we can generate by choosing a subset of the blocks and for each chosen block, flipping every pair inside. Now the two mixtures $M$ and $M'$ are defined as follows:\
\begin{enumerate}

\item[(1)] Set $M = \sum_{i = 1}^{2^{\ell-1}} \frac{1}{2^{\ell-1}} M(\phi, \pi_i)$ where $\pi_i$ are the $2^{\ell-1}$ permutations where for an even number of blocks, the paired elements are flipped. 

\item[(2)] Similarly set $M' = \sum_{i = 1}^{2^{\ell-1}}\frac{1}{2^{\ell-1}}  M(\phi, \pi'_i)$ where $\pi'_i$ are the $2^{\ell-1}$ permutations where for an odd number of blocks, the paired elements are flipped. 

\end{enumerate}

\noindent We will break up the proof into two parts, where we first show that the mixtures are non-degenerate. Then we establish a lower bound on the cost in the local query model. 

\begin{proof}[Proof of First Part]
We want to show that the components of the mixture are separated from the uniform distribution and from each other in total variation distance.  For any of the components in either of the mixtures, the first element of its base permutation maps to each of the positions $1,2, \dots , n$ with probabilities  $$\Big(\frac{1-\phi}{1-\phi^n}, \dots ,  \frac{(1-\phi)\phi^{n-1}}{1-\phi^n}\Big)$$  We will denote this vector by $v_n$.  Note $\phi^{\sqrt{\frac{n}{4k}}} > \frac{1}{2}$ and thus we can ensure that the first $(\frac{n}{4k})^{1/2}$ entries of $v_n$ are all at least $\frac{1}{2}(\frac{k}{n})^{1/2}$.  Let $e$ be the vector $(\frac{1}{n} , \dots , \frac{1}{n})$. Thus
$$\|v_n - e\|_1 \geq  \sqrt{\frac{n}{4k}}\Big(\frac{1}{2}\sqrt{\frac{k}{n}} - \frac{1}{n}\Big) > \frac{1}{5}$$  In particular, this implies that the total variation distance between any of the components and the uniform distribution is at least $\frac{1}{5}$.

Next we bound the total variation distance between any pair of components.  For every pair of components $M(\pi_1,\phi)$ and $ M(\pi_2,\phi)$ coming from either of the two mixtures, the permutations $\pi_1$ and $ \pi_2$ differ in that for some block, its pairs of consecutive elements are flipped in $\pi_2$.  Let $S$ be the set of all pairs of consecutive elements that are flipped between $\pi_1$ and $\pi_2$.  In our construction, we ensured that $|S| \geq \frac{n}{2\ell}$.  Let $|S| = m$.  To bound the total variation distance between $M(\pi_1,\phi)$ and $M(\pi_2,\phi)$, we look at which of the $m$ pairs in $S$ occur in reverse order (not necessarily consecutive) in a sample.  Since the pairs are disjoint, the probability that a sample from $M(\pi_1,\phi)$ has exactly $i$ of the pairs reversed is $\frac{\phi^i}{(1+\phi)^m}$ while the probability that a sample from $M(\pi_2,\phi)$ has exactly $i$ of them reversed is $\frac{\phi^{m-i}}{(1+\phi)^m}$.  Thus, the total variation distance between the distributions generated by $M_1,M_2$ is at least
$$ d_{TV}(M(\pi_1,\phi), M(\pi_2,\phi)) \geq \frac{1}{2}\frac{\sum_{i=0}^m \binom{m}{i}|\phi^i - \phi^{m-i}|}{(1+\phi)^m}$$
For $i \leq \frac{m-\sqrt{m}}{2}$ we have $|\phi^i - \phi^{m-i}| =\phi^i(1-\phi^{m-2i}) \geq  \phi^i(1-\phi^{\sqrt{m}}) \geq \phi^i(1-\phi^{\sqrt{\frac{n}{2l}}}) \geq \frac{1}{4}\phi^i$. Thus we conclude  
$$d_{TV}(M(\pi_1,\phi), M(\pi_2,\phi)) \geq \frac{1}{2}\frac{\sum_{i=0}^m \binom{m}{i}|\phi^i - \phi^{m-i}|}{(1+\phi)^m} \geq \frac{1}{4} \frac{\sum_{i=0}^{\frac{m-\sqrt{m}}{2}} \binom{m}{i}\phi^i}{(1+\phi)^m} \geq \frac{1}{4} \frac{\sum_{i=0}^{\frac{m-\sqrt{m}}{2}} \binom{m}{i}}{2^m} \geq \frac{1}{40}$$
It is easy to see that the same argument works for lower bounding $ d_{TV}(M(\pi_1,\phi), U)$ where $U$ is the uniform distribution on permutations. This completes the proof of the first part. 
\end{proof}

\begin{proof}[Proof of Second Part] Now we prove a lower bound on the cost to distinguish between $M$ and $M'$.  First, we observe that any local query involving strictly less than $\ell$ elements has the same answer for both mixtures.  Indeed, one of the $\ell$ blocks must have no queried elements and thus we can pair the components in $M$ with the components in $M'$ that are the same except that all pairs of consecutive elements in this block are flipped.  For each such pair of components, the distribution of the queried elements is the same and thus overall, the distributions of the queried elements in the even mixture and in the odd mixture must be the same.  

Next, we show that for both mixtures, the answer to any query involving at least $\ell$ elements can be answered with zero unless the tolerance parameter is at least $(\frac{2k}{n})^{\ell/2}$.  It suffices to upper bound the probability that some $\ell$ elements $i_1, \dots , i_l$ map to some $\ell$ positions $x_1, \dots , x_l$.  Let $p$ be the probability that $i_1$ maps to position $x_1$.  The probability that $i_1$ maps to an adjacent position is at least $\phi p$ since for any permutation with $i_1$ in position $x_1$, swapping $i_1$ with an adjacent element creates at most $1$ more inversion.  Repeating this argument, the probability that $i_1$ maps to a position $2$ away from $x_1$ is at least $\phi^2 p$ and so on.  Overall, we conclude that $p(1+\phi+ \dots + \phi^{n-1}) \leq 1$ and thus $p \leq \frac{1}{1+\phi+ \dots + \phi^{n-1}}$.  

We can repeat a similar argument for the probability that $i_2$ maps to potion $x_2$ conditional on $i_1$ mapping to position $x_1$.  In this case, the denominator has one less term since $i_2$ cannot map to $x_1$ and we can upper bound the probability by $\frac{1}{1+ \phi^2 + \dots + \phi^{n-1}} \leq \frac{1}{\phi+ \phi^2 \dots \phi^{n-1}}$.  Bounding the probabilities for each of the $\ell$ elements $i_j$ mapping to $x_j$, the overall probability that $i_1, \dots , i_l$ map to positions $x_1, \dots , x_l$ is at most 
\begin{align*}
\frac{1}{1+\phi+ \dots + \phi^{n-1}} \cdot \frac{1}{\phi+ \phi^2 \dots \phi^{n-1}} \dots \frac{1}{\phi^{l-1}+ \dots + \phi^{n-1}} = \frac{(1-\phi)^l}{(1-\phi^n)(\phi-\phi^n) \dots (\phi^{l-1} - \phi^n)} \\
= \bigg( \sqrt{\frac{k}{n}} \bigg)^l\frac{1}{(1-\phi^n)(\phi-\phi^n) \dots (\phi^{l-1} - \phi^n)} \leq \bigg( \sqrt{\frac{2k}{n}} \bigg)^l
\end{align*}
which completes the proof. 
\end{proof}

\section{Beyond Worst-Case Analysis}
Motivated by our lower bound against local algorithms, it is natural to ask whether there is some notion of beyond worst-case analysis whereby we can get much faster algorithms that work under tame conditions on the input. Here we give such an algorithm in the case when all of the scaling parameters are separated from each other (and from the value one, which causes a different sort of degeneracy). 

\begin{definition}
We say a mixture of Mallows models
$M = w_1 M(\phi_1,\pi_1) + \cdots + w_k M(\phi_k,\pi_k)$
is $(\gamma, \alpha)$-separated if $w_i \geq \alpha$ and $\phi_i \leq 1-\gamma$ for all $i$ and additionally $|\phi_i - \phi_j| \geq \gamma$.
\end{definition}

In our main lemmas, we will also need some new notions of component-wise closeness that will arise due to some subtleties at intermediate steps. 

\begin{definition}
We say that two mixtures of Mallows models
$M = w_1 M(\phi_1,\pi_1) + \cdots + w_k M(\phi_k,\pi_k)$ and $M' = w'_1 M(\phi'_1,\pi'_1) + \cdots + w'_k M(\phi'_k,\pi'_k)$ with the same number of components are component-wise $\theta$-close in parameters if there is a relabelling of components in one of the mixtures after which $|w_i - w'_i|, |\phi_i - \phi'_i| \leq \theta$ and $\pi_i = \pi'_i $ for all $i$. If all but the condition on the mixing weights holds, we will say that they are component-wise $\theta$-close in base parameters. 
\end{definition}

\begin{theorem} \label{smoothed}
Consider a mixture of $k$ Mallows models
$M = w_1 M(\phi_1, \pi_1) + \cdots + w_k M(\phi_k, \pi_k) $ that is $(\gamma, \alpha)$-separated and has $n \geq 10k$. There is an algorithm that runs in time polynomial in $n$, $\log 1/\delta$ and $(\frac{1}{\theta \gamma \alpha})^{k^2}$ time 
and outputs a mixture $M'$ that is component-wise $\theta$-close in parameters to $M$ with probability at least $1-\delta$. 
\end{theorem}

We give a brief outline of our algorithm. First we run the algorithm in Claim~\ref{prefixsep} to find candidate prefixes for all the base permutations and estimates of their scaling parameters. Next we run the algorithm in Lemma~\ref{prefixtofull} to extend the prefixes to a list of candidate full base permutations. Then we run the algorithm in Claim~\ref{nobrutemix} so that for every $k$ tuple of candidate base permutations and scaling parameters, we can uniquely determine accurate estimates of the mixing weights. Finally we run the algorithm in Lemma~\ref{testwithsep} to test whether a given mixture is component-wise close in parameters to the true mixture. 

\subsection{Finding the Prefixes}

Our first step is to determine the first $10k$ elements of each base permutation. More precisely, given samples from $M$, we want to find a list of candidate prefixes, so that for each base permutation its prefix appears on the list. Along the way, we will also perform a grid search over the scaling parameters. 

\begin{claim}\label{prefixsep}
Consider a mixture of $k$ Mallows models
$M = w_1 M(\phi_1, \pi_1) + \cdots + w_k M(\phi_k, \pi_k) $ that is $(\gamma, \alpha)$-separated. There is an algorithm that takes $$m = \Big( \frac{1}{2k  \alpha \gamma^{10k} \log \frac{1}{\delta}} \Big)^2$$ samples and runs in time polynomial in $n$, $1/(\alpha \beta)^k$, $1/\gamma^{k^2}$ and $\log 1/\delta$ and outputs a list of size $$s = 2^k\frac{1}{\gamma^ {10k^2}} \frac{1}{(\alpha \beta)^k }$$ where each entry is a candidate prefix of $10k$ elements and a scaling parameter. Moreover with probability at least $1-\delta$, for each $i$, the first $10k$ elements of $\pi_i$ appear as an entry in the list along with an estimate of $\phi_i$ that is close within an additive $\beta$. 
\end{claim}
\begin{proof}
The first $10k$ elements of a permutation $\pi_i$ appear as the first $10k$ elements of a random sample from the mixture with probability at least 
$$\frac{\alpha}{(1+\phi+\cdots + \phi^{n-1})^{10k}}  \geq (1-\phi)^{10k}\alpha \geq \gamma^{10k} \alpha$$
Given $m$ samples, with probability $1-\delta$, the prefixes of each permutation $\pi_i$ appears as the first $10k$ elements at least a $\frac{1}{2}\gamma^ {10k} \alpha$ fraction of the time. Now we can simply take all prefixes of $10k$ elements that appear at least a $\frac{1}{2}\gamma^{10k} \alpha$ fraction of the time and all choices of the scaling parameter in a $\beta$-grid.  This gives a list of size at most $s$ and completes the proof of the claim. 
\end{proof}

This is the only step in our algorithm which is not obvious how to implement using local queries. Nevertheless, it can be: First, query the probability that each of the $n$ elements appears as the first element in a draw from $M$. Now for each of the elements that occurs first with probability at least $\frac{1}{2}\gamma^ {10k} \alpha$, query the probability of all possible second elements (conditioned on the choice of the first element). We can repeat this process to find the heavy hitters among all possible prefixes of $10k$ elements, and can remove any prefix that does not occur with probability at least $\frac{1}{2}\gamma^ {10k} \alpha$ thus ensuring that the list of queries we need to make does not ever become too large.

\subsection{Finding the Full Permutations and Mixing Weights}

In this subsection, we show how to find all the base permutations from their prefixes. We will also show how to recover the mixing weights using the base permutations and scaling parameters. 

\begin{lemma}\label{prefixtofull}
Suppose the conditions of Theorem~\ref{smoothed} hold. Suppose that $n \geq 10k$ and the first $10k$ elements of each permutation are known and we are given estimates $\phi'_i$ of the scaling parameters that, for each $i$, satisfy $|\phi_i-\phi_i'| \leq \beta$ with $\beta \leq \alpha\big( \frac{\gamma}{10} \big)^{3k}$. There is an algorithm that runs in time polynomial in $n$, $1/\alpha$, $1/\beta$, $\log 1/\delta$ and $1/\gamma^k$ , uses $m = \big( 2^kn  \frac{1}{\beta} \log \frac{1}{\delta} \big)^2$ samples, and outputs a list of $2^{k-1}$ permutations for each $i$ so that with probability at least $1-\delta$, each $\pi_i$ is included on the corresponding list. 
\end{lemma}
\begin{proof}
To simplify the exposition, we first give an algorithm that works when the scaling parameters $\phi_i$ and moments of the distribution are exactly known. We will then show that it continues to work even when we use the estimates $\phi'_i$ instead. Our algorithm will recover each base permutation $\pi_i$ separately. Consider the first $2k-2$ elements of the last permutation $x_1,x_2, \cdots, x_{2k-2}$ and two additional elements $x$ and $y$ that we wish to recover the relative order of, in the first permutation. Now consider the set $\mathcal{S}$ of permutations where the first $2k-2$ elements are $x_1,x_2, \cdots, x_{2k-2}$ in that order except up to possibly reversing the order  of some pairs $x_{2a-1}$ and $x_{2a}$. We can now form a $2 \times 2 \times \cdots \times 2$ tensor $T_{k,x}$ of order $k-1$ where
$$T_{k,x} = p_{k,x} v(M(\phi_k, \pi_k \vert_{\{x_1, x_2\}})) \otimes \cdots \otimes v(M(\phi_k, \pi_k \vert_{\{x_{2k-3}, x_{2k-2}\}}))$$
where $p_{k,x} = \Pr_{M(\phi_k, \pi_k)}[\pi \in \mathcal{S} \mbox{ and } \pi(x) < \pi(y)]$.
We can analogously define $T_{k,y}$ but where we require $\pi(y) < \pi(x)$. Finally we set 
$$T_x = \sum_i w_i T_{i,x} \mbox{ and } T_y = \sum_i w_i T_{i,y}$$

Now we guess, for each $a$ from $1$ to $k-1$, the relative ordering of $x_{2a-1}$ and $x_{2a}$. As usual we can construct test functions for isolating a single component. In particular, if $x_{2a-1}$ occurs before $x_{2a}$ in $\pi_a$ set $v_a = (\frac{\phi_a}{1+\phi_a}, \frac{-1}{1+\phi_a})$. Otherwise set $v_a = (\frac{1}{1+\phi_a}, \frac{-\phi_a}{1+\phi_a})$. It is easy to see that $v_a$ is orthogonal to $v(M(\phi_a, \pi_a \vert_{\{x_{2a-1}, x_{2a}\}}))$ and furthermore
$$| \langle v_a, v(M(\phi_b, \pi_b \vert_{\{x_{2b-1}, x_{2b}\}})) \rangle | \geq \frac{|\phi_a - \phi_b|}{4}$$
Now set
$Z = v_1 \otimes v_2 \otimes \cdots \otimes v_{k-1}$
and it is easy to see that if $x$ occurs before $y$ in $\pi_k$ we have
$$|\langle Z, T_x \rangle| = |\langle Z, T_{k,x} \rangle| \geq \Big (\frac{1}{\phi_k}\Big ) |\langle Z, T_y \rangle| = |\langle Z, T_{k,y} \rangle|$$
and if $y$ occurs before $x$, the same inequality holds but with the roles of $x$ and $y$ interchanged. Thus we can decide if $x$ occurs before $y$ in $\pi_k$ by comparing the values of $\langle Z, T_x \rangle$ and $\langle Z, T_y \rangle$. Moreover we can repeat this procedure for every pair of elements (using the same guess) to recover the entire base permutation $\pi_k$. 

Now we remove the assumption that the scaling parameters are known, and also bound the error introduced by estimating the tensors $T_x$ and $T_y$ from samples. First, if we take  $m = \big( 2^kn  \frac{1}{\beta} \log \frac{1}{\delta} \big)^2$ samples then with probability at least $1 - \frac{\delta}{kn^2}$ we have that $\|T_x - T'_x\|_1, \|T_y - T'_y\|_1 \leq \beta$ where $T'_x$ and $T'_y$ are the obvious empirical estimates. Next we construct $Z'$ using the $\phi'_i$'s instead of the $\phi_i$'s. It is easy to see that $\|Z-Z'\|_1 \leq k 2^k \beta$. Finally we have
$$\Big| |\langle Z, T_x \rangle| - |\langle Z, T_y \rangle| \Big| \geq \frac{1-\phi_k}{1+\phi_k}w_k\Big(p_{k,x}+p_{k,y}\Big)\Big(\frac{\gamma}{4}\Big)^{k-1}$$
and moreover
$$p_{k,x}+p_{k,y} \geq \frac{1}{(1+ \phi_k + \cdots + \phi_k^{n-1})^{2k}} \geq \gamma^ {2k}$$
which follows because $p_{k,x}+p_{k,y}$ is at least the probability that the first $2k-2$ elements of $\pi_k$ occur in order, for a draw from $M(\phi_k, \pi_k)$. Putting it all together we have
$$\Big| |\langle Z, T_x \rangle| - |\langle Z, T_y \rangle| \Big| \geq \alpha \Big ( \frac{\gamma}{4} \Big)^{k-1}$$
Now since we set $\beta \leq \alpha\big( \frac{\gamma}{10} \big)^{3k}$, it follows that we can compare the values of $\langle Z', T'_x \rangle$ and $\langle Z', T'_y \rangle$ to recover the relative order of $x$ and $y$ in $\pi_k$ with failure probability at most $\frac{\delta}{n^2k}$.  Now union bounding over the $n^2k$ steps (when our guess is correct), we conclude that our algorithm succeeds with probability at least $1-\delta$.
\end{proof}

It might seem like we can now just grid search over the mixing weights. But there is a subtle issue: If all the base permutations and scaling parameters are correct, but the mixing weights are not, our testing algorithm might not be able to tell. For this reason we need to make sure that once we have a candidate set of base permutations and their scaling parameters, we do not need to do any guessing to determine the mixing weights. The following claim shows how the ideas in the proof of Lemma~\ref{prefixtofull} can be adapted to resolve this issue. 

\begin{claim}\label{nobrutemix}
Suppose the conditions of Theorem~\ref{smoothed} hold. Furthermore suppose the base permutations are known and we are given estimates $\phi'_i$ of the scaling parameters that, for each $i$, satisfy $|\phi_i-\phi_i'| < \beta$ with $\beta \leq \alpha^2\big( \frac{\gamma}{10} \big)^{6k}$. There is an algorithm that runs in time polynomial in $n$, $1/\beta$ and $\log 1/\delta$ time and uses $m = \big( 2^k n  \frac{1}{\beta} \log \frac{1}{\delta} \big)^2$ samples and outputs estimates of the mixing weights that satisfy $$|w_i - w'_i| \leq  \frac{\beta}{\alpha}\Big( \frac{10}{ \gamma} \Big)^{4k}$$ for each $i$, with probability at least $1-\delta$. 
\end{claim}

\begin{proof}
Recall we defined $T_{k,x}$ and $T_k$ in Lemma~\ref{prefixtofull}. Here we will use a variant of these constructions. In particular, $T_k$ has the same definition as $T_{k,x}$ except that (for the same definition of $\mathcal{S}$) we replace $p_{k,x}$ with $p_k = \Pr_{M(\phi_k, \pi_k)}[\pi \in \mathcal{S}]$. Similarly let $T = \sum_i w_i T_i$. Now we can use $m$ samples to compute $T'$ with the property $\|T - T' \|_1 \leq \beta$. We can also compute $Z'$ with $\|Z-Z'\|_1 \leq k 2^k \beta$, but this time without any guessing because we are assuming that the base permutations are all known. Finally note that the entries of $T_k$ only depend on $\pi_k$ and $\phi_k$. Let $T'_k$ be the result of replacing $\phi_k$ (which we do not know) with $\phi'_k$. 

Next we bound the difference between $T_k$ and $T'_k$: Note that each entry of $T_k$ is of the form
$$\frac{\phi_k^i}{(1+\phi_k + \cdots + \phi_k^{n-1}) \cdots (1+ \cdots + \phi_k^{n-2k+2})}$$
where $i$ is the number of pairs $x_{2a-1}$ and $x_{2a}$ (from the $2k-2$ prefix of $\pi_k$) whose order is exchanged. In $T'_k$, the entries are of the same form but with $\phi_k$ replaced with $\phi'_k$. Suppose $\phi_k \leq \phi'_k$. Then if $a$ and $a'$ are two corresponding entries in $T_k$ and $T'_k$ respectively we have
$$\frac{a}{a'} \leq \frac{(1+\phi_k' + \cdots + \phi_k'^{n-1}) \cdots (1+ \cdots + \phi_k'^{n-2k+2})}{(1+\phi_k + \cdots + \phi_k^{n-1}) \cdots (1+ \cdots + \phi_k^{n-2k+2})} \leq \frac{(1-\phi_k)^{2k-2}}{(1-\phi_k')^{2k-2}} \leq (1+\frac{\beta}{\gamma})^{2k-2} \leq 1+\frac{4k\beta}{\gamma}$$
for $\beta \leq \frac{1}{2k \gamma}$. Thus it follows that the total difference between $T_k$ and $T'_k$ over entries where the former is larger, is at most $ \frac{4k \beta}{\gamma}$ since the sum of entries of $T'_k$ is one. Similarly the difference between $T_k$ and $T'_k$ over entries where the latter is larger is at most
$$(1+\phi_k')^{k-1} - (1+\phi_k)^{k-1} \leq (k-1)(1+\phi_k')^{k-2} (\phi_k'-\phi_k) \leq k2^k\beta$$
Now we can estimate the mixing weight of the $k$th component as
$$w'_k = \frac{\langle T', Z' \rangle}{\langle T'_k, Z' \rangle}$$
Note that $\langle T, Z \rangle = w_k \langle T_k, Z \rangle$. Also $w_k \geq \alpha$ and $$|\langle T_k, Z \rangle| \geq p_k \Big(\frac{\gamma}{4}\Big)^{k-1} \geq \Big(\frac{\gamma}{4}\Big)^{3k}$$.

Now note that for real numbers $x,y,x',y'$ with $|x|, |y|, |x'|, |y'| \leq 1$
$$|xy - x'y'| = |x(y-y') + (x-x')y'| \leq |y-y'| + |x - x'|$$
Since all entries of $T_k, T, Z, T_k', T', Z'$ are less than $1$, we have
$$| \langle T_k, Z \rangle| - | \langle T_k', Z' \rangle| \leq \| T_k - T_k' \|_1 + \| Z - Z' \|_1 \leq k2^k \frac{\beta}{\gamma}$$

$$| \langle T, Z \rangle| - | \langle T', Z' \rangle| \leq \| T - T' \|_1 + \| Z - Z' \|_1 \leq (k2^k + 1) \beta$$

Putting everything together and using our assumption about $\beta$, we have
$$|w_k - w_k'| = \Big| \frac{\langle T', Z' \rangle}{\langle T'_k, Z' \rangle} - \frac{\langle T, Z \rangle}{\langle T_k, Z \rangle}\Big|  \leq \frac{ |(k2^k + 1)\beta \langle T_k, Z \rangle| + |k2^k\frac{\beta}{\gamma}\langle T, Z \rangle|}{|\langle T_k, Z \rangle \langle T_k', Z' \rangle|} \leq \frac{4k2^k\frac{\beta}{\gamma}}{\alpha (\frac{\gamma}{4})^{3k}} \leq  \frac{\beta}{\alpha}\Big( \frac{10}{ \gamma} \Big)^{4k}$$
which completes the proof.
\end{proof}

\subsection{Testing Closeness for Separated Mixtures}

The final piece of our algorithm is a method for testing if two $(\frac{\gamma}{2}, \alpha)$-separated mixtures $M$ and $M'$ are component-wise close in parameters.

\begin{lemma}\label{testwithsep}
Suppose we are given sample access to a mixture of $k$ Mallows models $M = w_1 M(\phi_1, \pi_1) + \cdots + w_k M(\phi_k, \pi_k) $ and an estimate $M' = w'_1 M(\phi'_1, \pi'_1) + \cdots + w_k M(\phi'_k, \pi'_k) $ and both are $(\frac{\gamma}{2}, \alpha)$-separated. Finally suppose $n \geq 10k$ and $\theta \leq \frac{\gamma}{10}$. There is an algorithm which given 
$$m = \Big( \frac{nk 10^k \log \frac{1}{\delta}}{\theta \alpha  \gamma^k} \Big)^{20}$$
samples from $M$ runs in polynomial in $n$, $1/\gamma^k$, $1/\alpha$, $1/\theta$ and $\log 1/\delta$ time and if $M$ and $M'$ are not component-wise $\theta$-close in base parameters, rejects with probability at least $1-\delta$. And if they are component wise $\theta'$-close in parameters, for
$$\theta' = \big( \frac{\theta \alpha \gamma^k}{10^k} \big)^{50}$$
it accepts with probability at least $1-\delta$. 
\end{lemma}
\begin{proof}
Suppose that $M$ and $M'$ are not component-wise $\theta$-close in base parameters. From the fact that $\theta \leq \frac{\gamma}{10}$, there is a component, say $M(\phi'_1, \pi'_1)$ that is not $\theta$-close (i.e. has the same base permutation and scaling parameter within an additive $\theta$) to any component of $M$. Suppose that $M(\phi_1, \pi_1)$ is the component of $M$ whose scaling parameter is the closest to $\phi'_1$.  We will break into two cases depending on whether $|\phi_1 - \phi'_1| \leq \theta$.

\textit{Proof of First Case} First consider when $|\phi_1 - \phi'_1| \leq\theta$. Then there must be a pair of elements, say $x$ and $y$, whose order in $\pi_1$ is different from their order in $\pi'_1$. We will suppose without loss of generality that $x$ and $y$ are not in the first $4k$ elements of $\pi'_1$ (as otherwise they will not be in the last $4k$ elements, and we can repeat the entire argument but globally flipping all the orderings).  Now suppose the first $4k-4$ elements of $\pi'_1$ are $x_1,x_2, \cdots, x_{4k-5},x_{4k-4}$. We will use a variant of a construction we have used many times already, e.g. in Lemma~\ref{prefixtofull} where we consider the set $\mathcal{S}$ of permutations where these first $4k-4$ elements appear first, except up to flipping pairs $x_{2a-1}$ and $x_{2a}$. For each component $M(\phi_i, \pi_i)$ we can define the tensor $T_{i, x}$ to be the $2 \times 2 \times \cdots \times 2$ order $2k-2$ tensor where
$$T_{i,x} = p_{i,x} v(M(\phi_i, \pi_i \vert_{\{x_1, x_2\}})) \otimes \cdots \otimes v(M(\phi_i, \pi_i \vert_{\{x_{4k-5}, x_{4k-4}\}}))$$
And as usual $p_{i,x} = \Pr_{M(\phi_i, \pi_i)}[\pi \in \mathcal{S} \mbox{ and } \pi(x) < \pi(y)]$. Define $T_{i, y}$ analogously but with the roles of $x$ and $y$ interchanged. Also let $T'_{i,x}$ and $T'_{i,y}$ denote the corresponding tensors using the components of $M'$ instead. And finally let 
$$T_x = \sum_i w_i T_{i,x}, T_y = \sum_i w_i T_{i,y}, T'_x = \sum_i w'_i T'_{i,x} \mbox{ and } T'_y = \sum_i w'_i T'_{i,y}$$
Each of these tensors can be approximated by sampling from $M$ or $M'$. Now if $x_{2a-1}$ occurs before $x_{2a}$ in $\pi_i$ set $v_{i,a} = (\frac{\phi_i}{1+\phi_i}, \frac{-1}{1+\phi_i})$. Otherwise set $v_{i,a} = (\frac{1}{1+\phi_i}, \frac{-\phi_i}{1+\phi_i})$. Similarly we can define $v'_{i,a}$. With this definition in hand, we set
$$Z = v_{2, 1} \otimes \cdots \otimes v_{k, k-1} \otimes v'_{2, k-2}, \otimes \cdots \otimes v'_{k, 2k-2}$$
By construction (if, say $x$ occurs before $y$ in $\pi_1$) we have that $|\langle Z, T_x \rangle | > |\langle Z, T_y \rangle |$ and $|\langle Z, T'_x \rangle | < |\langle Z, T'_y \rangle | $. Repeating essentially the same calculations as in Lemma~\ref{prefixtofull} it is easy to see that the gap between these two inequalities is at least $\alpha\big( \frac{\gamma}{8} \big)^{6k}$. Using the largest entry in absolute value in $Z$ is at most one, we conclude that if we take $m$ samples then either $\|T_x - T'_x\|_1$ or $\|T_y - T'_y\|_1$ will be at least  $\frac{1}{3}\alpha\big( \frac{\gamma}{8} \big)^{6k}$ in which case we reject. It is easy to see that the failure probability in this case is at most $\frac{\delta}{10kn^2}$.

Now we must show that if $M$ and $M'$ are component wise $\theta'$-close in parameters then $\|T_x - T'_x\|_1$ and $\|T_y - T'_y\|_1$ will both be less than $\frac{1}{4}\alpha\big( \frac{\gamma}{8} \big)^{6k}$ with failure probability at most $\frac{\delta}{10kn^2}$.

We will bound each error $\|T_{i,x} - T'_{i,x}\|_1$ separately.  WLOG $\phi_i < \phi_i'$ and $x$ is ranked ahead of $y$ in $\pi_i$.  Let $X = [n] \backslash \{ x_1, x_2, \dots , x_{4k-4} \}$.    Consider an entry of $T_{i,x}$.    The value can be expressed in the form
$$
\frac{\phi_i^{c_i}}{(1+ \phi_i+ \dots + \phi_i^{n-1}) \dots (1+ \phi_i+ \dots + \phi_i^{n-(4k-4)})} \Pr_{M(\phi_i, \pi_i|_X)}[ \pi(x) < \pi(y)]
$$ 
where $c_i$ is some integer.   We will bound the error incurred by replacing $\phi_i$ with $\phi_i'$ in each of the above terms.  First note that 
$$
\begin{multlined}
\phi_i^{c_i} - \phi_i'^{c_i} \leq \\
\frac{\phi_i^{c_i}}{(1+ \phi_i+ \dots + \phi_i^{n-1}) \dots (1+ \phi_i+ \dots + \phi_i^{n-(4k-4)})} - \frac{\phi_i'^{c_i}}{(1+ \phi_i'+ \dots + \phi_i'^{n-1}) \dots (1+ \phi_i'+ \dots + \phi_i'^{n-(4k-4)})} \\ 
\leq \frac{(1+ \phi_i'+ \dots + \phi_i'^{n-1}) \dots (1+ \phi_i'+ \dots + \phi_i'^{n-(4k-4)})}{(1+ \phi_i+ \dots + \phi_i^{n-1}) \dots (1+ \phi_i+ \dots + \phi_i^{n-(4k-4)})} - 1 \leq \frac{(1-\phi_i')^{4k-4}}{(1 - \phi_i)^{4k-4}} - 1
\end{multlined}
$$
and $|\phi_i'^{c_i} - \phi_i^{c_i}| \leq c_i\phi_i'^{c_i - 1}$.  The maximum of $c_i\phi_i'^{c_i - 1}$ occurs when $c_i = \frac{1}{-\log \phi_i'}$ and thus
$$
|\phi_i'^{c_i} - \phi_i^{c_i}| \leq \frac{1}{-\log \phi_i'} \leq \frac{1}{1 - \phi_i'}
$$
Now, we show an explicit method to compute $\Pr_{M(\phi_i, \pi_i|_X)}[ \pi(x) < \pi(y)]$.  We note that if $x$ ranked exactly $d$ elements ahead of $y$ in $\pi_i|X$ then 
$$
\Pr_{M(\phi_i, \pi_i|_X)}[ \pi(x) < \pi(y)] = \Pr_{M(\phi_i, (1,2, \dots d))}[ \pi(1) < \pi(d)]
$$
The last expression can be expressed as
$$
\frac{\sum_{1 \leq r < s \leq d}\phi^{r-1 + d-s}}{(1+ \phi_i + \dots + \phi_i^{d-1})(1+ \phi_i + \dots + \phi_i^{d-2})}
$$
since each term in the numerator represents the probability that $1$ maps to position $r$ and $d$ maps to position $s$. Note that $\Pr_{M(\phi_i, (1,2, \dots d))}[ \pi(1) < \pi(d)] \geq  \Pr_{M(\phi_i', (1,2, \dots d))}[ \pi(1) < \pi(d)]$.  However,
$$
\frac{\Pr_{M(\phi_i, (1,2, \dots d))}[ \pi(1) < \pi(d)]}{\Pr_{M(\phi_i', (1,2, \dots d))}[ \pi(1) < \pi(d)]} \leq \frac{(1+ \phi_i' + \dots + \phi_i'^{d-1})(1+ \phi_i' + \dots + \phi_i'^{d-2})}{(1+ \phi_i + \dots + \phi_i^{d-1})(1+ \phi_i + \dots + \phi_i^{d-2})} \leq \frac{(1- \phi_i')^2}{(1 - \phi_i)^2}
$$
and since clearly $\Pr_{M(\phi_i', (1,2, \dots d))}[ \pi(1) < \pi(d)] \leq 1$,
$$
|\Pr_{M(\phi_i, (1,2, \dots d))}[ \pi(1) < \pi(d)] - \Pr_{M(\phi_i, (1,2, \dots d))}[ \pi(1) < \pi(d)]  | \leq \frac{(1- \phi_i')^2}{(1 - \phi_i)^2} - 1
$$
Combining all of the above inequalities, we can bound $\|T_{i,x} - T'_{i,x}\|_1$ entrywise.  Similarly, we can bound $\|T_{i,y} - T'_{i,y}\|_1$ and altogether, we can check that with our choice of $\theta'$, if the mixtures $M$ and $M'$ are $\theta'$-close then $\|T_{i,x} - T'_{i,x}\|_1, \|T_{i,y} - T'_{i,y}\|_1 < \frac{1}{4}\alpha\big( \frac{\gamma}{8} \big)^{6k}$ with failure probability at most $\frac{\delta}{10kn^2}$.  This completes the proof of the first case.

\textit{Proof of Second Case} The second case where $|\phi_1 - \phi'_1| > \theta$ can be handled similarly: We construct a $2 \times 2 \times \cdots \times 2$ order $2k-1$ tensor $T$ based on the first $4k-2$ elements (i.e. we do not  use $x$ and $y$ at all, as we did above). We can then construct a $Z$ so that $\langle Z, T \rangle = 0$ but $$|\langle Z, T' \rangle| \geq \alpha \theta \Big( \frac{\gamma}{8} \Big)^{6k}$$
Again using $m$ samples we can reject if $\|T - T'\|_1 \geq |\langle Z, T' \rangle| - |\langle Z, T \rangle|\geq \frac{1}{2}\alpha \theta \Big( \frac{\gamma}{8} \Big)^{6k}$ which will fail with probability at most $\frac{\delta}{10kn^2}$. Similar to the first case, it is easy to see that if $M$ and $M'$ are component wise $\theta'$-close in parameters then $\|T' - T\|_1 \leq  \frac{1}{4}\alpha \theta \Big( \frac{\gamma}{8} \Big)^{6k}$ with failure probability at most $\frac{\delta}{10kn^2}$, in which case $|\langle Z, T' \rangle|  - |\langle Z, T \rangle|  \leq  \frac{1}{4}\alpha \theta \Big( \frac{\gamma}{8} \Big)^{6k}$ for all $Z$ where the largest entry in absolute value is at most one. This completes the proof. 
\end{proof}

We are now ready to complete the proof of Theorem~\ref{mainsmooth}:

\begin{proof}
First we run the algorithm from Claim~\ref{prefixsep} to generate a list of candidate prefixes. Then for each $k$ tuple of prefixes, we run the algorithm from Lemma~\ref{prefixtofull} which in turn generates a list of candidate base permutations. Then we consider all $k$ tuples of candidate base permutations and do a grid search of side length $\beta = \big( \frac{\theta \alpha \gamma^k}{10^k} \big)^{100}$ over the possible scaling parameters. For $k$ tuple of candidate base permutations and estimates for the scaling parameters, we run the algorithm from Claim~\ref{nobrutemix} to estimate the mixing weight of each component. This generates a list of candidate mixtures, at least one of which is component-wise $\theta' = \big( \frac{\theta \alpha \gamma^k}{10^k} \big)^{50}$ close in parameters to the true mixture $M$. 

We then run the testing algorithm in Lemma~\ref{testwithsep} on each candidate mixture with the failure probability $\delta'$. Then the total failure probability is at most $ 10 \delta' ( \frac{2^k}{\alpha \beta \gamma^k} )^{10k}$. We can set $\delta'$ appropriately to make this quantity be at most $\delta$. Now if the testing algorithm does not accept a mixture that should have been rejected, or vice-versa, it outputs at least one mixture $M'$ which must be component-wise $\theta$-close in base parameters to the true mixture. From the guarantees of  Claim~\ref{nobrutemix} we know that its mixing weights must also be close to the true mixing weights, which finally completes the proof of correctness.
\end{proof}

\bibliographystyle{plain}

\appendix 

\section{Omitted Proofs} \label{appendix}

Here we give the proofs deferred from Section~\ref{TVdistanceBounds}. First we prove Claim~\ref{simple}:

\begin{proof}
Assume without loss of generality that $\phi_2 \geq \phi_1$. Now let $x$ and $y$ be a pair of elements where $x$ is ranked higher than $y$ in $\pi_1$, and $x$ is ranked lower than $y$ in $\pi_2$. Then the total variation distance between $M(\phi_1, \pi_1)$ and  $M(\phi_2, \pi_2)$ is at least the difference between the probabilities they rank $x$ higher than $y$, which is
$$\frac{1}{1+\phi_1}- \frac{\phi_2}{1+\phi_2} \geq \frac{\epsilon}{2}$$
which completes the proof. 
\end{proof}

Next we prove Lemma~\ref{TVbound}:

\begin{proof}
Consider each permutation $\pi'$. Let $d = d_{KT}(\pi', \pi)$ and set $Z_i(\phi) = 1 + \phi + \cdots + \phi^{i-1}$. The probability of generating $\pi'$ under $M_1$ is 
\[
\Pr_{M_1}[\pi'] = \frac{\phi_1^d}{Z_1(\phi_1) \cdots Z_n(\phi_1)} = \frac{\phi_1^d(1-\phi_1)^{n-1}}{(1-\phi_1^2) \cdots (1-\phi_1^n)}
\]
while the probability of generating it under $M_2$ is 
\[ 
\Pr_{M_2}[\pi'] = \frac{\phi_2^d}{Z_1(\phi_2) \cdots Z_n(\phi_2)} = \frac{\phi_2^d(1-\phi_2)^{n-1}}{(1-\phi_2^2) \cdots (1-\phi_2^n)}
\]
First if $\phi_i < \frac{\mu}{2(n-1)}$ then the probability of generating $\pi$ is at least $(1-\frac{\mu}{2(n-1)})^{n-1} \geq 1-\frac{\mu}{2}$ so if both $\phi_1, \phi_2 < \frac{\mu}{2(n-1)}$ then the total variation distance between the distributions is at most $\mu$.

Now it suffices to consider the case where both $\phi_i \geq \frac{\mu}{2n}$.  Without loss of generality $\phi_1 \geq \phi_2$.  We will bound the ratio between the two probabilities above.  Note $\phi_1^d \geq \phi_2^d$ and $\frac{(1-\phi_1)^{n-1}}{(1-\phi_1^2) \cdots (1-\phi_1^n)} \leq \frac{(1-\phi_2)^{n-1}}{(1-\phi_2^2) \cdots (1-\phi_2^n)}$.  We have
$$\frac{\Pr_{M_1}[\pi']}{\Pr_{M_2}[\pi']} \leq \Big (\frac{\phi_1}{\phi_2}\Big)^d \leq \Big (1+\frac{\frac{\mu^2}{10n^3}}{\phi_2} \Big)^d \leq \Big(1+\frac{\mu}{5n^2}\Big)^d \leq \Big(1+\frac{\mu}{5n^2}\Big)^{n^2} \leq 1+\frac{\mu}{2}
$$
where the last inequality is true when $n \geq 1, \mu <1$ (which is effectively always).  Also
\[
\frac{\Pr_{M_1}[\pi']}{\Pr_{M_2}[\pi']} \geq \frac{1+\phi_2}{1+\phi_1} \cdots  \frac{1+\phi_2+ \cdots +\phi_2^{n-1}}{1+\phi_1+ \cdots +  \phi_1^{n-1}}
\]
We now bound each term separately
\begin{eqnarray*}
\frac{1+\phi_2 + \cdots +\phi_2^{i-1}}{1+\phi_1 + \cdots +\phi_1^{i-1}} &=& 1 - \frac{(\phi_1-\phi_2)+(\phi_1^2 - \phi_2^2)+ \cdots (\phi_1^{i-1}-\phi_2^{i-1})}{1+\phi_1 + \cdots +\phi_1^{i-1}} \\
& \geq& 1- \Big (\frac{\mu^2}{10n^3} + 2\frac{\mu^2}{10n^3} + \cdots + (i-1)\frac{\mu^2}{10n^3}\Big ) \geq 1-\frac{\mu^2}{10n}
\end{eqnarray*}

Therefore we conclude
$$\frac{\Pr_{M_1}[\pi']}{\Pr_{M_2}[\pi']} \geq \Big (1-\frac{\mu^2}{10n}\Big)^{n-1} \geq 1- \frac{\mu}{2}$$
Combining both bounds on the ratio $\frac{\Pr_{M_1}[\pi']}{\Pr_{M_2}[\pi']}$, we see that the sum of  $|\Pr_{M_1}[\pi]-\Pr_{M_2}[\pi]|$ over $\pi'$ with $\Pr_{M_1}[\pi'] \geq \Pr_{M_2}[\pi']$ is at most $\frac{\mu}{2}$ and similarly for the sum over $\pi'$ with $\Pr_{M_1}[\pi'] < \Pr_{M_2}[\pi']$.  Thus the total variation distance between the two distributions is at most $\mu$, completing the proof.
\end{proof}

\end{document}